\documentclass[fleqn,10pt]{wlscirep}
\usepackage[utf8]{inputenc}
\usepackage[T1]{fontenc}

\usepackage{amsmath,amsthm,amssymb}
\usepackage{epsfig,minibox,multirow,makecell,tikz-cd,comment,xcolor,graphics}
\usepackage{changepage}
\usepackage{hyperref}      
\usepackage{thmtools}
\usepackage{amsfonts}       
\usepackage{graphicx}
\usepackage{tikz}
\usepackage{mathtools}
\usetikzlibrary{matrix,positioning}
\usepackage[caption=false]{subfig} 
\usepackage{caption}
\usepackage{thm-restate}

\usepackage{siunitx}
\usepackage{booktabs}
\usepackage{array}
\usepackage{csvsimple}
\usepackage{adjustbox} 
\usepackage{soul}
\newcolumntype{L}[1]{>{\raggedright\arraybackslash}p{#1}}

\usepackage{rotating}
\usepackage{fix-cm}
\usepackage{relsize}
\tikzset{fontscale/.style = {font=\relsize{#1}}
    }

\usetikzlibrary{decorations.markings}
\usetikzlibrary{calc}

 \usepackage{tikz}
 \usetikzlibrary{arrows}
 \usepackage{listings}

\usepackage{kotex}

\theoremstyle{plain}

\newtheorem{theorem}{Theorem}[section]

\newtheorem{proposition}[theorem]{Proposition}

\theoremstyle{definition}

\theoremstyle{remark}
\newtheorem{remark}[theorem]{Remark}

\title{Zipf–Mandelbrot Scaling in Korean Court Music: Universal Patterns in Music}

\author[1,+]{Byeongchan Choi}
\author[1,+]{Junwon You}
\author[2]{Myung Ock Kim}
\author[1,*]{Jae-Hun Jung}
\affil[1]{POSTECH, Department of Mathmatics, Pohang, 37673, Korea}
\affil[2]{Korea Institute for Advanced Study (KIAS), Seoul, 02455, Korea}

\affil[*]{corresponding author jung153@postech.ac.kr}

\affil[+]{these authors contributed equally to this work}


\begin{abstract}
Zipf’s law, originally discovered in natural language and later generalized to the Zipf–Mandelbrot law, describes a power-law relationship between the frequency of a Zipfian element and its rank. Due to the semantic characteristics of this law, it has also been observed in musical data. However, most such studies have focused on Western music, and its applicability to non-Western music remains not well investigated. We analyzed 43 Korean court music pieces called Jeong-ak, spanning several centuries and written in the traditional Korean musical notation Jeongganbo. These pieces were transcribed into Western staff notation, and musical data such as pitch and duration were extracted. Using pitch, duration, and their paired combinations as Zipfian units, we found that Korean music also fits the Zipf–Mandelbrot law to a high degree, particularly for the paired pitch–duration unit. Korean music has evolved collectively over long periods, smoothing idiosyncratic variations and producing forms that are widely understandable among people. This collective evolution appears to have played a significant role in shaping the characteristics that lead to the satisfaction of Zipf-Mandelbrot law. Our findings provide additional evidence that Zipf–Mandelbrot scaling in musical data is universal across cultures.
We further show that the joint distribution of two independent Zipfian data sets follows the Zipf–Mandelbrot law; in this sense, our result does not merely extend Zipf’s law but deepens our understanding of how scaling laws behave under composition and interaction, offering a more unified perspective on rank-based statistical regularities.
\end{abstract}

\begin{document}

\flushbottom
\maketitle
%
%
\thispagestyle{empty}


\section{Introduction}\label{sec:Introduction}

Zipf's law originally discovered in natural language, describes a power-law relationship between the frequency $f, f \in \mathbb{Z}_{>0}$ and rank $r, r \in \mathbb{Z}_{>0}$ of elements known as the Zipfian elements within a dataset \cite{zipf1935psycho, zipf1949human}, given as 
\begin{equation}
    f(r) \sim \frac{1}{r^s}    
\end{equation}
where $s, s \in \mathbb{R}_{>0}$ is the exponent, i.e. the slope of the power-law decay of $f$ on a logarithmic scale. The original finding that led to Zipf's law suggests that the slope of the power decay is approximately $-1$, i.e., $s \sim 1$. 
This slope value $s$ is related to the semantic structure of a dataset. 
The slope $s \approx 1$ is often interpreted as an indicator of ``efficiency" in balancing redundancy and diversity in communication. In natural language, $s \approx 1$ has been connected to cognitive or semantic structures, optimal information transfer, and communicative efficiency. Thus, by analogy, if another similar dataset found in such as music, ecology, genetics, etc. shows a similar slope value of $s$, it could imply that it also reflects some ``semantic" or ``structural" property of the system that generates the dataset \cite{zipf1935psycho, zipf1949human,mandelbrot1953informational,ferrer2003least}. 
As Zipf's law indicates, the occurrence frequency of ranked elements decays algebraically with increasing rank. Consequently, the most frequent and the rarest elements in a dataset exhibit a large disparity in frequency. This pattern was initially explained in the context of human behavior and psychology, particularly in terms of communicative efficiency and the principle of least effort \cite{zipf1935psycho, zipf1949human}. 
Subsequent research demonstrated that Zipf's law is not limited to natural language dataset but is found in a wide range of empirical data. Examples include city population distributions \cite{hackmann2020zipf}, website visitation frequencies \cite{adamic2002zipf}, and species abundances in ecology \cite{frank2019log}. While the ubiquity of Zipf’s law is remarkable, the underlying reasons for its widespread emergence still remain incompletely understood. 

The original Zipf's law assumes an idealized infinite system and overlooks finite-size effects. In real-world datasets, the most frequent elements often deviate from a pure power-law distribution. To account for such deviations, the Zipf's law was generalized into the Zipf–Mandelbrot law \cite{mandelbrot1953informational, mandelbrot1965information, herdan1960type}, given as 
\begin{equation}
    f(r) \sim \frac{1}{(r+q)^s},    
\end{equation}
where $r$ is the rank, $s$ is the exponent as in Zipf’s law, and $q, q\in\mathbb{R}_{\ge 0}$ is an additional parameter, known as the shift parameter, that shifts the rank to account for finite-size effects and high-frequency deviations. It was shown that natural language data, and later other complex systems such as music, city sizes, and income distributions, deviate systematically from the ideal Zipf's law \cite{mandelbrot1953informational, mandelbrot1965information}. The shift parameter $q$ enables more accurate modeling of empirical frequency distributions. For the top few ranks, when $r$ is small, the effect of $q$ is significant, flattening the curve. For large $r$, $r+q \approx r$, the law reduces to the standard Zipf form. This modification 
provides a better fit for real datasets.

Given its close association with semantic structure, the original Zipf's law has been applied to characterize statistical patterns in music datasets, particularly those of Western music. In Western music, numerous studies have investigated both the original Zipf's law and its Zipf–Mandelbrot extension. For example, Voss and Clarke \cite{voss1978noise} found that musical signals exhibit $1/f$ noise characteristics, with pitch fluctuations following an inverse power-law distribution. Here $f$ is the temporal frequency of fluctuations in the musical signals such as pitch or amplitude changes. This finding positions music between pure randomness and strict periodicity, producing a balance between predictability (order) and surprise (randomness) that is perceptually pleasing. The power-law slope $s \approx 1$ ensures that energy is distributed across scales, giving a fractal-like structure, which implies that music is neither fully random nor fully repetitive. Manaris et al. \cite{manaris2002zipfs} applied the Zipf–Mandelbrot law to a large corpus of 220 MIDI-encoded works spanning Baroque, Classical, Romantic, 12-tone, jazz, rock, and algorithmically generated music, showing that these metrics capture aesthetic properties computationally. 
The metrics used in this work include pitch, duration, melodic intervals, and harmonic consonance. These findings indicate that metrics following a Zipfian distribution give rise to music with aesthetic significance. This implies a connection between Zipf’s law and the musical semantics that listeners intuitively use to understand and enjoy music. 

The application of Zipf’s law to music is not straightforward, as it requires defining the semantic elements or Zipfian units analogous to words in natural language. However, in music, the identification of such semantic elements is not well established, and various studies have tested different choices. Perotti and Billoni \cite{perotti2020emergence} found that Zipfian distributions appear  when both chords and notes are considered as Zipfian units, reinforcing the analogy between music and language. 
Zanette \cite{zanette2004zipfs} extended the concept of linguistic context to music and showed that the statistics of note usage follow Simon’s model, a generative mechanism underlying Zipf’s law. Further refinements have also been proposed. 
Here Simon’s model \cite{simon1955class} is a stochastic model for generating distributions that follow Zipf’s law. It was originally proposed to explain word frequency distributions in natural language, but it has broader applications. The model generates easily a Zipfian distribution with a simple stochastic process. In \cite{zanette2004zipfs}, musical notes or events are treated as elements in a sequence, and Simon’s model is applied to explain how musical contexts and note usage can lead to Zipf-like statistics; the probability of repeating a note is proportional to how often it has appeared in the musical sequence, leading to a power-law distribution in note frequencies. This implies that music like natural language may exhibit ``semantic context", where some notes or motifs are repeated more often in certain contexts. 
Serra-Peralta et al. \cite{serra2021heaps} analyzed harmonic codewords from classical music using the Kunstderfuge corpus of classical music\footnote{https://www.kunstderfuge.com/}.
They used Heap's law, $V(N) \sim K N^\beta$ where $V(N)$ is the number of distinct elements (or the vocabulary size), $N$ is the total number of elements in a dataset, $K >0$ is a constant and $\beta, 0<\beta<1$ is the exponent. Then they showed the mathematical connection between the vocabulary size and Zipf's law, summarized by $\beta \approx \frac{1}{s}$. That is, the authors took distinct harmonic codewords as musical vocabulary and showed such codewords satisfy Zipf's law. 
Haro et al. \cite{haro2012zipf} 
analyzed timbral codewords extracted from 740 hours of speech, music, and environmental sounds and showed that Zipfian distributions are ubiquitous in auditory signals, extending beyond natural language to music and environmental sounds. This result support the idea that Zipf’s law is a general feature of structured symbolic sequences, including musical timbres, not only word sequences in language.

Despite the extensive research on Zipf’s law in music, its application has primarily focused on Western music. As discussed above, Zipf’s law yields a significant implication that the construction of music may rely on underlying musical semantics, providing a statistical framework that listeners intuitively use to understand and appreciate compositions. Empirical evidence from previous studies supports this idea, making it highly relevant for musicological and aesthetic analyses. However, such evidence is largely limited to Western music, and it still remains unclear whether similar patterns hold for other musical traditions, particularly in datasets where formal music theory is less established or underrepresented. Unlike Western music, which has well-developed theoretical frameworks such as harmony and counterpoint, many ethnomusic traditions lack such formalized systems.

In this context, we analyze a large-scale dataset of Korean traditional court music,  {\it Jeong-ak} to investigate its statistical distributions and assess whether it also follows Zipf’s law. Unlike Western compositions, Korean traditional music was not composed by a formal unified music theory.  Instead, each piece appears to follow its own compositional principles. This makes statistical analysis of Korean music more challenging, as the structure is less codified and more variable compared to Western music.

As explained above, it is important to define the semantic elements, or Zipfian units, of a dataset to apply Zipf's law. Following previous works such as \cite{voss1978noise, perotti2020emergence, zanette2004zipfs}, we consider pitch, duration, and their paired combinations as the Zipfian units for Korean traditional music. Korean traditional music is essentially heterophonic 
and does not incorporate the concept of harmony. Therefore, in this study, a Zipfian analysis based on harmonic structure is not applicable.

For this study, we analyzed 43 {\it Jeong-ak} music pieces written in the traditional Korean music notation {\it Jeongganbo}, spanning several centuries, with the oldest dating back approximately 600 years. Although a larger corpus of Korean traditional music exists, these pieces were selected because they are preserved in written form, unlike many orally transmitted works. {\it Jeong-ak} music was predominantly performed in the royal court and among the aristocracy, and its documentation in notation enables accurate statistical analysis. As detailed in Section 2, we translated the {\it Jeongganbo} scores into Western staff notation to facilitate compatibility with standard computational music analysis tools, followed by careful preprocessing and normalization.
The KIAS Large-scale Korean Music Research Group converted {\it Jeongganbo} notation into Western staff notation. During this process, the pitch and duration conversions were applied differently according to several musical conventions, so the data were normalized to combine them into a single dataset for analysis\cite{kim2023gugakAI2}. 


Our analysis demonstrates that Korean traditional court music follows the Zipf–Mandelbrot law across all Zipfian units, including pitch, duration, and their combinations. Particularly, the paired unit of pitch and duration provides clear evidence that Korean music follows the same statistical pattern observed in Western music. Moreover, our results suggest that Korean music exhibits an even \textcolor{black}{stronger adherence to the Zipf–Mandelbrot distribution compared to Western music}. This can be explained by the collective and long-term nature of Korean traditional music, which has evolved through group creativity rather than individual composition. 

As noted by Zipf \cite{zipf1935psycho, zipf1949human}, Mandelbrot \cite{mandelbrot1953informational, mandelbrot1965information}, Voss and Clarke \cite{voss1978noise}, Serra-Peralta et al. \cite{serra2021heaps} and others, the presence of Zipf’s law in music suggests that musical sequences are structured around semantic units, enabling listeners to perceive and appreciate them similarly to natural language. The result we obtained that Korean traditional music also follows Zipf’s law indicates a comparable underlying semantic structure. Moreover, its stronger adherence suggests that over time, Korean music has been collectively shaped to be more easily understood by listeners, reflecting a more communal rather than individual compositional process. That is, over centuries, idiosyncratic effects have been smoothed out, leading to forms shaped by collective cultural consensus. Such long-term, large-scale, and self-organizing processes are precisely the conditions under which Zipf-like distributions typically emerge, as in natural language. In contrast, works by a single composer or short-lived musical styles may lack the scale and diversity necessary for stable Zipfian patterns. Corral et al. \cite{corral2015zipf} showed that Zipf’s law is more robust in long, collective texts than in short or single-author samples, emphasizing the stabilizing effect of scale and cultural accumulation. Similarly, Serra-Peralta et al. \cite{serra2021heaps} found that Zipfian scaling in music emerges more clearly in large and diverse corpora rather than isolated works.

We further show that the joint distribution of two independent Zipfian data sets follows the Zipf–Mandelbrot law. This result shows how scaling laws behave under composition and interaction, and provides a more unified perspective on rank-based statistical regularities. This observation also applies to the Korean music data considered in this study. In particular, for the joint distribution of pitch and duration, each marginal data set satisfies at least a piecewise Zipf’s law. Since most data from nature arise as joint distributions of two or more variables, this result may provide insight into why the Zipf–Mandelbrot law frequently emerges in empirical data.

This paper is organized as follows. In Section 2, we describe the dataset used in this study. Specifically, the dataset was derived from large-scale Jeongganbo music scores, first translated into Western staff notation and then subjected to preprocessing and normalization. In this section, we also define the Zipfian units employed in our analysis. In Section 3, we present the Zipf–Mandelbrot law and discuss the implications of its parameters. This section also explains the methods used to estimate these parameters. Section 4 presents a detailed data analysis and demonstrates that the data follow the Zipf–Mandelbrot distribution.
Section 5 provides a theoretical discussion, including results on the Zipf–Mandelbrot law for paired Zipfian data.
Finally, in Section 6, we provide our conclusions and discuss directions for future research.


\section{Dataset and Zipfian Units}
\label{sec:Method}

Korean court music has been passed down in the form of notation in \textit{Jeongganbo}, created by King Sejong(1397-1450).  
In \textit{Jeongganbo}, pitch is written in Chinese characters 
called \textit{yulmyeong}. However, a more fundamental feature of \textit{Jeongganbo} is that it allows us to recognize the duration of the sound by applying it to series of squares.

We analyzed the data of Korean music constructed by the KIAS Large-scale Korean Music Research Group. This data consists of 45 court music pieces from \textit{Jeongganbo}. The selected pieces represent the standard repertoire currently performed at the National Gugak Center in Korea and the \textit{Jeongganbo} was published by the National Gugak Center in 2015-2016. Since \textit{Jeongganbo} is typically published as separate parts rather than a full score, the data were converted into 220 MusicXML files, with each file representing an individual instrumental part. 

All the pieces in this analysis are ensemble works composed generally between the 15th and 19th centuries. The instruments analyzed are typical Korean court instruments \textit{daegeum}, \textit{piri}, \textit{haegeum}, \textit{gayageum}, \textit{geomungo}, \textit{ajaeng} and voice. Among the 45 pieces in this data, only movements 1 to 3 out of the 7 movements of \textit{Yeomillak} are included and \textit{Gagok}  was originally performed as a large suite of 41 pieces, but this data includes only three pieces. We perform this data in our study after removing eight files consisting of two pieces, \textit{Boheoja} and \textit{Haeryeong}, that are influenced by Chinese music. Chinese music has been imported into Korea over an extended period; however, its musical style differs from that of Korean music. Since the 15th century, several pieces have been created in the Chinese musical style. Accordingly, musicologists distinguish these works from original Korean music, given their Chinese stylistic features. On this basis, in this study 43 pieces of 212 musicXML files from the data were used. In addition, in this dataset, \textit{Yeomillak} and \textit{Yangcheong Dodeori} were divided into multiple files; the criteria are given later. The complete list of musical pieces considered in this paper is provided in Table 3.

\begin{figure}[!ht]
    \centering
    \includegraphics[width=0.8\linewidth]{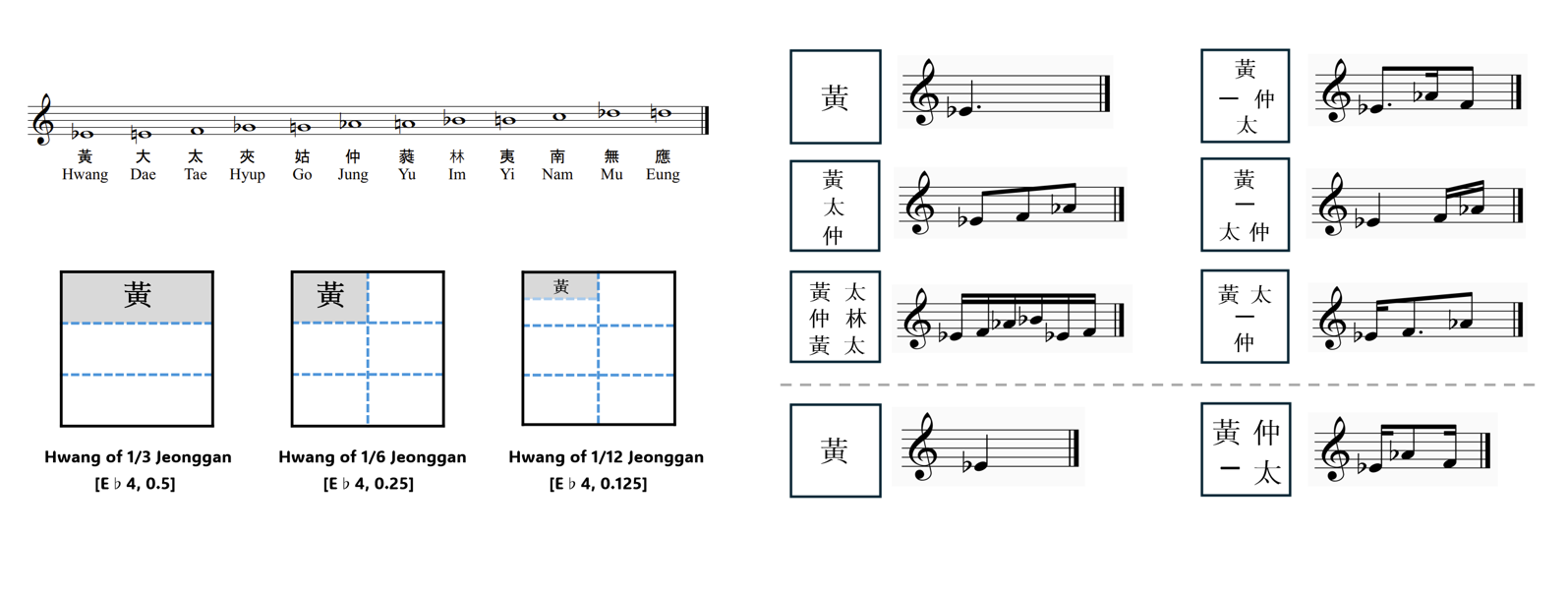}
    \caption{Interpretation of \textit{Jeongganbo} notation. Korean traditional pitch names (\textit{yulmyeong}) are mapped to corresponding Western pitches, with Hwang (黃) corresponding to E$\,\flat$ 4. Duration is represented by the spatial structure of the \textit{Jeonggan}: one square cell indicates one beat, which may be subdivided into two or three smaller units and further subdivided recursively. The placement and grouping of pitch symbols within these divisions specify both pitch and duration, allowing \textit{Jeongganbo} to be interpreted and transcribed into Western staff notation.}
    \label{fig:Interpretation of Jeongganbo 02}
\end{figure}

 The pitch names called \textit{yulmyeong} were invented in ancient China and appear in Korean \textit{Jeongganbo} after the 15th century. 
 There are twelve \textit{yulmyeong} in one octave and these correspond to specific Western notes (see Fig.~\ref{fig:Interpretation of Jeongganbo 02}). In Korean music, the keynote \textit{hwang}(黃) corresponds to  E$\,\flat$ and each of the twelve pitches has its own name. \textit{Jeonggan} is one square cell of the \textit{Jeongganbo} notation.
 The space occupied by the \textit{yulmyeong} in the \textit{Jeonggan} indicates its duration. Thus, the \textit{yulmyeong} was converted into pitch, and the \textit{Jeonggan} area was converted into duration. The Sejong-style \textit{Jeongganbo} was improved into a modern-style \textit{Jeongganbo} in the 20th century and is still used in performances today. The modern-style \textit{Jeongganbo} contains more detailed musical information. A single \textit{Jeonggan} indicates one beat and can be divided into three or two small beats, generally depending on the piece. These divided beats can then be divided into two again repeatedly (see Fig.~\ref{fig:Interpretation of Jeongganbo 02}). In this way, \textit{Jeonggan} clearly indicates pitch and duration and can be converted to staff notation.


Korean traditional music is known for its heterophony, where the melody of each part is considered more important than the vertical aspect. 
In a single piece of music, the fundamental melodic line was shared, but the melody of each part has been created by an independent entity. This process took place over decades or even centuries. Therefore, the melody of each part has a more substantial relevance to the process of creating this music than the vertical relationship between the parts. 
For this reason, we analyzed the data of each part in the piece. 



\subsection{Data Preprocessing and Definition of Zipfian Units}

To examine the applicability of Zipf's law to Korean traditional music, 
\textit{Jeongganbo} notation was first converted into the Western staff notation and stored it in MusicXML format~\cite{good2001musicxml}. 
The converted staffs were then processed using the \texttt{music21} python library~\cite{cuthbert2010music21}, which provides robust tools for parsing, analyzing, and extracting symbolic music information.  

Each MusicXML file was loaded and parsed using the \texttt{music21.converter.parse()} function. From the parsed data, we extracted \textit{pitch} and \textit{duration} information from individual musical events. Non-pitched elements (e.g., rests, slurs, lyrics) were excluded from the analysis. Only valid \texttt{Notes} 
were retained: a \texttt{Note} represents a single pitch event with a specific duration.

Korean music has evolved from simple to delicate melodies. 
This long process of forming a piece played an important role in melody formation; therefore, \textit{Jeongganbo} encodes various ornaments in addition to pitch and duration. However improvisation still applies to this music, and musical elements such as ornamentation may change depending on the performer. Accordingly, \textit{yulmyeong} and duration are the most fundamental elements for melody formation in \textit{Jeongganbo}.
(To examine) the semantic structure of melodies, especially from the perspective of creation, only pitch and duration were considered.


\begin{figure}[!ht]
    \centering
    \includegraphics[width=0.7\linewidth]{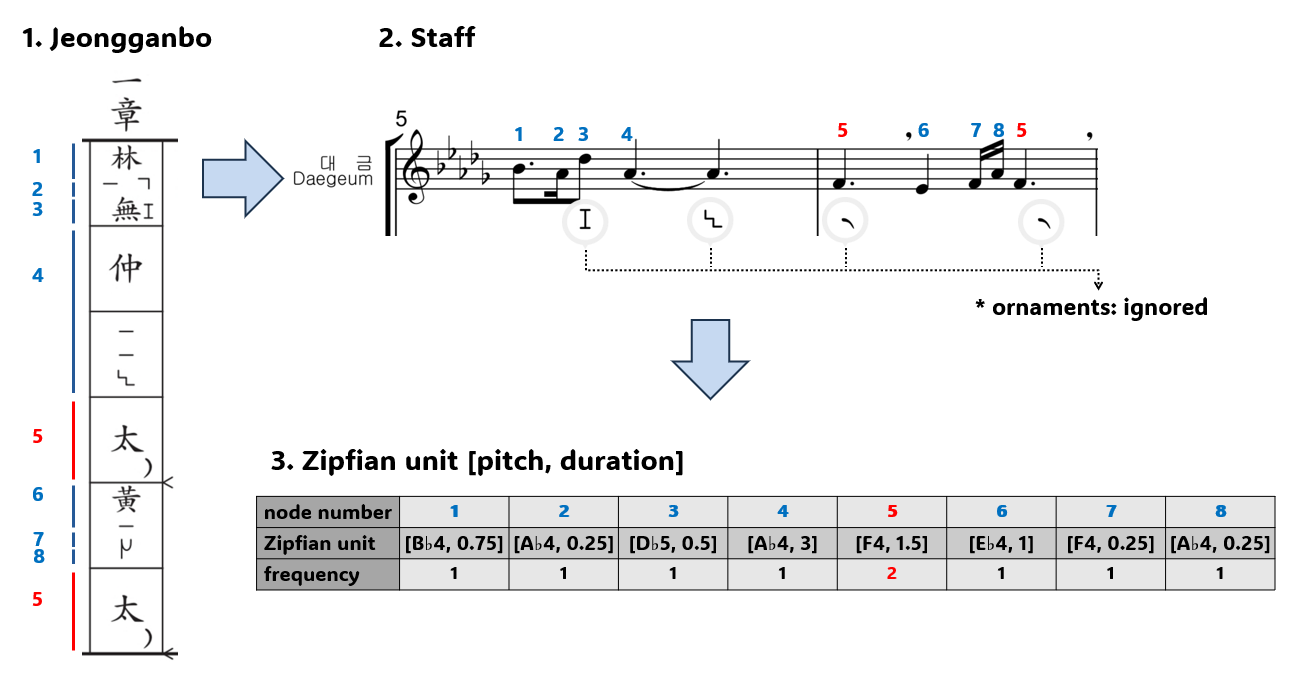}
    \caption{Overview of the process for converting Jeongganbo into Zipfian units. Jeongganbo is first transcribed into Western staff notation, enabling computational analysis of Korean court music using the \texttt{music21} library. Subsequently, \texttt{(pitch, duration)} pairs are extracted as Zipfian units, and their frequencies are computed. Ornamentation symbols are ignored.}
    \label{fig:Zipfian unit_J-Sanghyun_daegeum}
\end{figure}


In this study, we converted traditional \textit{Jeongganbo} pitch names into western note names 
(see Fig.~\ref{fig:Zipfian unit_J-Sanghyun_daegeum}).  
The mapping was determined by referencing common scholarly interpretations and performance practices of Korean court music. For instance, the note \textit{Hwang} (黃) was mapped to E$\,\flat$ 4, \textit{Tae} (太) to F 4, and so on. Fig.~\ref{fig:yulmyeong table& scale} (in Appendix) summarizes the correspondence used in this study.


Each note was represented as a tuple consisting of its note name (e.g., E$\,\flat$ 4, F 4) and rhythmic duration, measured in \textit{Jeonggan} units. We defined these \texttt{(pitch, duration)} tuples as \textit{Zipfian units}, capturing both melodic and rhythmic dimensions of the musical structure. 

To ensure comparability across pieces composed for different instruments, we performed pitch normalization based on the typical pitch ranges of each instrument. Since the dataset included part scores played on a variety of instruments, such as \textit{daegeum}, \textit{piri}, \textit{haegeum}, \textit{gayageum}, \textit{geomungo}, and \textit{ajaeng}, pitch normalization was necessary to align all musical events to a consistent register.

In Korean court music, the concept of a “central range” has long been recognized. Consistent with this, a study of medieval Korean music found that note occurrences were highest within that central octave range~\cite{Hwang1999yonganhoesang}.


Although modern instruments can produce a wider pitch range than their medieval counterparts, musical usage still tends to concentrate on each instrument’s specific register~\cite{Chun2013instrument}.
So we normalized the most frequent pitch range of the instruments. 
Specifically, pitches played on the \textit{daegeum}, which naturally sounds in a higher register, were transposed down by one octave. 
In contrast, pitches from string instruments like the \textit{gayageum}, \textit{geomungo}, and \textit{ajaeng}, which often sound in a  lower register in performance, were transposed up by one octave. 
Instruments such as the \textit{piri} and \textit{haegeum}, which typically operate in the central range, were left unaltered. 
This approach enabled a meaningful comparison of pitch patterns across diverse instrumental settings.



During the conversion of \textit{Jeongganbo} notation into MusicXML format, we first established a base temporal unit by assuming that one \textit{Jeonggan} (i.e., one square cell of the notation) corresponds to 1.5 quarter notes in Western rhythmic units. This working assumption provided an initial reference for mapping durations across pieces.

However, in Korean music, the temporal interpretation of a \textit{Jeonggan} is not fixed and varies depending on the musical context. As described above, a single \textit{Jeonggan} may represent either a unit of three small beats or a unit of two small beats. In court music, the three-beat structure is the most common (see Fig.~\ref{fig:Interpretation of Jeongganbo 02}), but variations occur across musical genres and pieces.

Within the analyzed corpus, \textit{Yangcheongdodeuri} is the only piece that employs both two-beat and partially three-beat definitions of the \textit{Jeonggan} unit. Because these two definitions correspond to different absolute durations, \textit{Yangcheongdodeuri} was divided into two separate files according to the underlying beat structure, ensuring that each file maintained a consistent interpretation.

Another source of temporal variability arises from \textit{Jangdan}, a rhythmic framework defined by a fixed number of \textit{Jeonggan}, for which multiple types exist. \textit{Taryeong-jangdan} is one such type, and in modern \textit{Jeongganbo} notation its duration is conventionally doubled. To recover the intended temporal scale, the durations of the pieces employing \textit{Taryeong-jangdan} were scaled by a factor of $\times$0.5. In contrast, \textit{Yeomillak} and several other pieces employ the same three-beat \textit{Jeonggan} unit, yet are notated on the Western staff as quarter notes with two-thirds of the standard duration. To restore proportional consistency, the durations in these pieces were scaled by a factor of $\times$1.5. These discrepancies indicate that the notated durations reflect conventional practices rather than absolute temporal values.

Therefore, to preserve the internal rhythmic proportionality of each piece while enabling cross-piece comparison, we applied piecewise duration scaling. Specifically, each piece was scaled by a factor of $\times$0.5, $\times$1, or $\times$1.5, depending on the tempo and expressive characteristics of its rhythmic structure. This duration normalization ensured that the temporal structure of each work remained internally coherent while allowing for consistent comparative analysis across the dataset.


The preprocessing pipeline consisted of the following steps:

\begin{itemize}
    \item \textbf{Loading MusicXML Files:} All scores were parsed using \texttt{music21.converter.parse()}.
    
    \item \textbf{Filtering Non-Musical Elements:} Only \texttt{Note} and \texttt{Chord} objects were retained, while rests, ties, and other metadata were excluded.
    
    \item \textbf{Extracting Pitch and Duration:} From each valid note, the \textit{note name} (e.g., C4, D4) and its corresponding duration (in \texttt{quarterLength}) were extracted. In the case of chords, each individual pitch was treated as a separate note. To enable cross-instrument and cross-piece comparison, both pitch and duration values were normalized. Pitches were transposed according to the typical register of each instrument (e.g., \textit{daegeum}: $-1$ octave, \textit{gayageum}, \textit{geomungo}, and \textit{ajaeng}: $+1$ octave, \textit{piri} and \textit{haegeum}: unchanged). For duration, although one \textit{Jeonggan} was generally mapped to 1.5 quarter notes, a piecewise scaling factor ($\times$0.5, $\times$1, or $\times$1.5) was applied based on the rhythmic characteristics (\textit{jangdan}) of each composition.
    
    \item \textbf{Constructing Zipfian Units:} Each \texttt{(pitch, duration)} pair was treated as a unique symbolic unit, and the frequency of each unit across the dataset was computed.
\end{itemize}

These Zipfian units served as the fundamental elements for rank-frequency analysis. By constructing a ranked frequency distribution of the extracted units, we evaluated the extent to which Korean traditional music exhibits Zipfian behavior.

\section{Zipf-Mandelbrot law}\label{subsec:ZMLaw}
In this section, we explain the Zipf-Mandelbrot law and the implications of its parameters, followed by the methodologies we adopt to estimate their values.

\subsection{Zipf-Mandelbrot law and its slope in the log-log plane}
Zipf's law is an empirical statistical principle that describes how the frequency of an element relates to its rank with respect to frequency.
First observed by George Kingsley Zipf in the study of natural language~\cite{zipf2016human}, this law has since been identified in a wide variety of natural and social phenomena.
Mathematically, Zipf's law states that the frequency $f(r), f \in \mathbb{Z}_{>0}$ of the element with rank $r, r \in \mathbb{Z}_{>0}$ follows the relation
\begin{equation}
    f(r) \propto \frac{1}{r^{s}},
\end{equation}
where $s, s \in \mathbb{R}_{>0}$ is the scaling exponent and is often close to $1$ in empirical datasets.
On a log-log plot of frequency versus rank, data following Zipf's law form an straight line with slope $-s$.
In music, the “elements” or ``Zipfian units"  in Zipf's law can correspond to symbolic units such as pitch classes, rhythmic patterns, or combined representations like the \texttt{(pitch, duration)} units.
However, empirical data often deviate from the ideal Zipfian form, particularly for high-frequency elements with low ranks, resulting in curvature in the rank–frequency plot.
To address these deviations, the Zipf–Mandelbrot law introduces a shift parameter $q, q\in\mathbb{R}_{\ge 0}$~\cite{cannon1984fractal}, modifying the formula to
\begin{equation}
    f(r) \propto \frac{1}{(r + q)^{s}}.
\end{equation}
The Zipf–Mandelbrot law reduces to the classical Zipf's law when $q=0$.
In this study, we fit the Zipf–Mandelbrot law to the frequency distributions of normalized Zipfian units derived from Korean music scores.

For Zipf's law, writing $f(r) = A\, r^{-s}, A >0$ yields
    $\displaystyle\frac{d \log f}{d \log r} = -s,    $
so the slope on the log–log plot is the constant as $-s$.
By contrast, for the Zipf--Mandelbrot law $f(r) = A\,(r+q)^{-s}$ we have
    $\displaystyle\frac{d \log f}{d \log r} = 
     -\,s\,\frac{r}{r+q}, $
that is, the local slope depends on the rank $r$, whereas the slope in Zipf’s law is universal. Since the local slope varies with the rank $r$, this model is capable of capturing the flatter behavior observed in the high-frequency, low-rank region compared to the overall distribution. Taking limits clarifies the two regimes:
\begin{align}
\lim_{r \to \infty} \frac{d \log f}{d \log r} &= -s, \quad \text{for } r \gg q\\
\frac{d \log f}{d \log r} &\approx -\,s\,\frac{r}{q} \quad \text{for } r \ll q.
\end{align}
Thus, for large ranks ($r \gg q$) the local slope approaches $-s$, and the Zipf--Mandelbrot model becomes indistinguishable from the classical Zipf law in the tail. 
For small ranks ($r \ll q$), the slope magnitude is reduced to approximately $s\,r/q$, producing a flatter ``head'' and visible curvature in the rank-frequency plot. 
The bend occurs around $r \approx q$, so larger $q$ values yield a more pronounced and extended head curvature, while preserving the Zipfian slope $-s$ asymptotically.

\subsection{Paramter estimates}

Building on the rank–frequency behavior identified in Section 3.1, we now quantify how closely the data distributions are described by the Zipf–Mandelbrot model.
For each Zipfian unit, we estimate the parameters $(q,s)$ for a normalized fit and $(A,q,s)$ for a raw fit by fitting the theoretical rank–frequency curve
\begin{equation}
\label{eq:fit}
\begin{cases}
    f_{q,s}(r) = \dfrac{1}{(r + q)^{s}} & \text{for a normalized fit,}\\
    f_{A, q,s}(r) = \dfrac{A}{(r + q)^{s}} & \text{for a raw fit,}
\end{cases}
\end{equation}
to the observed data.
For the raw fit, we use the rank–frequency curve of the data.
In contrast, we fit the rank–\emph{normalized}-frequency curve, where the frequencies are normalized such that their sum across all ranks equals 1.

For finding the optimal parameter values of each model, we used the function \texttt{scipy.optimize.curve\_fit}, which is a function in the SciPy~\cite{virtanen2020scipy} library in Python used for performing non-linear least squares fitting of a user-defined function to data. 
The function returns optimal parameter values $(q,s)$ and $(A,q,s)$ so that the sum of the squared residuals of $f_{q,s}(r) - \hat{y}$ and $f_{A,q,s}(r) - y$ are minimized where $\hat{y}$ is a normalized frequency and $y$ is a frequency of Zipfian units in the observed data.

We used the Trust Region Reflective algorithm, which is a classical method that guarantees both local and globel convergence.~\cite{branch1999subspace, byrd1988approximate} 
The Trust Region Reflective algorithm is an optimization method used for nonlinear least-squares problems, particularly when there are bounds on the parameters. It works by iteratively defining a ``trust region” around the current estimate -- i.e., a region within which the model is trusted to approximate the objective function well -- and then finding a step that reduces the residuals within this region. After each step, the region may expand or contract depending on how well the model predicted the improvement. In \texttt{scipy.optimize.curve\_fit}, the Trust Region Reflective algorithm is used automatically when bounds are provided for the parameters (via the bounds argument). If no bounds are given, the default algorithm is Levenberg-Marquardt, which does not handle parameter bounds.
For our case, bounds are set to be $0 \le q \le 1000$, $0 \le s \le 20$ to prevent the overflow.

\section{Data Analysis}\label{sec:Data_Analysis}
\subsection{Summary of statistics}\label{subsec:summary-stats}

We begin by summarizing the most frequent units in our dataset (Table~\ref{tab:top-units}) to provide basic descriptive statistics. Throughout this section, the duration of each note is measured in units where the value $1.5$ corresponds to a single \emph{jeonggan}. The total number of units is $65{,}817$, so we report relative frequencies in percentage, rounded to one decimal place. Table~\ref{tab:top-units} summarizes the most frequent pitch, duration, and \texttt{(pitch, duration)} units in our corpus.

\paragraph{Pitch.}
The highest-ranked pitches concentrate in a single octave around a central range: the top seven pitch types all lie in this range, indicating that the pitch distribution is concentrated on a small set of tones within the same octave.
Among them, E$\flat$4 and B$\flat$4 appear most frequently and are typically realized as sustained \emph{keynotes}, whereas F4 occurs frequently as an \emph{intermediary tone} that connects adjacent keynotes.

\paragraph{Duration.}
The duration units are dominated by 1.5, 0.5, and 1.0.
Under our convention, 1.5 corresponds to one \emph{jeonggan}.
Since court music most commonly realizes a \emph{jeonggan} with a three-beat structure (see Fig.~\ref{fig:Interpretation of Jeongganbo 02}), the values 0.5 and 1.0 naturally arise as one- and two-beat subdivisions of a \emph{jeonggan}, which helps explain their high frequencies.

\paragraph{\texttt{(Pitch, duration)}.}
This structure carries over to the joint units: the most frequent \texttt{(pitch, duration)} pairs include (E$\flat$4,\,1.5) and (B$\flat$4,\,1.5), together with (F4,\,0.5).
In particular, the prominence of (F4,\,0.5) is consistent with the tendency for intermediary tones to occur with shorter durations, while keynotes more often appear with longer, sustained durations of a whole \emph{jeonggan}.

\begin{table}[!t]
  \centering
  \caption{Top-ranked pitch, duration, and \texttt{(pitch, duration)} units.
  Ranks are computed independently within each category.}
  \label{tab:top-units}
  \begin{adjustbox}{max width=0.65\linewidth}
  \begin{tabular}{r l r r l r r l r r}
    \toprule
    \multirow{2}{*}{Rank}
      & \multicolumn{3}{c}{Pitch}
      & \multicolumn{3}{c}{Duration}
      & \multicolumn{3}{c}{(Pitch, Duration)} \\
    \cmidrule(lr){2-4} \cmidrule(lr){5-7} \cmidrule(lr){8-10}
      & Unit & Count & Ratio(\%)
      & Unit & Count & Ratio(\%)
      & Unit & Count & Ratio(\%) \\
    \midrule
    1 & E$\flat$4 & 10177 & 15.5 & 1.5  & 14297 & 21.7 & E$\flat$4,\,1.5 & 2645 & 4.0 \\
    2 & B$\flat$4 & 9884  & 15.0 & 0.5  & 14272 & 21.7 & F4,\,0.5        & 2404 & 3.7 \\
    3 & F4        & 9440  & 14.3 & 1.0  & 10351 & 15.7 & B$\flat$4,\,1.5 & 2390 & 3.6 \\
    4 & A$\flat$4 & 8848  & 13.4 & 0.25 & 7706  & 11.7 & A$\flat$4,\,1.5 & 2025 & 3.1 \\
    5 & B$\flat$3 & 5588  & 8.5  & 3.0  & 5985  & 9.1  & B$\flat$4,\,1.0 & 1966 & 3.0 \\
    6 & C4        & 3953  & 6.0  & 0.75 & 3160  & 4.8  & B$\flat$4,\,0.5 & 1965 & 3.0 \\
    7 & C5        & 3024  & 4.6  & 1/12 & 2319  & 3.5  & F4,\,1.5        & 1690 & 2.6 \\
    8 & E$\flat$5 & 2973  & 4.5  & 2.5  & 2231  & 3.4  & E$\flat$4,\,0.5 & 1682 & 2.6 \\
    9 & A$\flat$3 & 2159  & 3.3  & 1/4  & 1775  & 2.7  & A$\flat$4,\,1.0 & 1546 & 2.3 \\
    10& G4        & 2088  & 3.2  & 1/3  & 1391  & 2.1  & A$\flat$4,\,0.5 & 1533 & 2.3 \\
    \bottomrule
  \end{tabular}
  \end{adjustbox}
\end{table}

\begin{table}[!t]
\centering
\small
\caption{Zipf--Mandelbrot fit results at the union and instrument levels.}
\label{tab:zm_fit_summary}

\begin{minipage}{\linewidth}
\centering
\textbf{Panel A: Union-level fits.}\par\smallskip
\begin{tabular}{lcccccccc}
\toprule
Unit & Fit & $A$ & $q$ & $s$ & $R^2$ & $N$ & $bar_{\min}$ & $bar_{\max}$ \\
\midrule
(p,d) & raw  & 66856.82 & 41.46 & 2.64 & 0.9980 & 555 & 18 & 34 \\
(p,d) & norm & 1.00     & 43.88 & 2.76 & 0.9979 & 555 & 18 & 33 \\
(p)   & raw  & 69222.83 & 101.61 & 20.00 & 0.9620 & 30  & 5  & 6 \\
(p)   & norm & 1.00     & 96.90  & 20.00 & 0.9600 & 30  & 5  & 6 \\
(d)   & ---  & ---      & ---    & ---   & ---    & --- & --- & --- \\
\multicolumn{9}{l}{\footnotesize\emph{Duration-only fits did not converge under the same fitting protocol.}}\\
\bottomrule
\end{tabular}
\end{minipage}

\vspace{1em}

\begin{minipage}{\linewidth}
\centering
\textbf{Panel B: Instrument-level fits.}\par\smallskip
\begin{adjustbox}{max width=\linewidth}
\csvreader[
  head to column names,
  separator=comma,
  tabular={
    L{0.18\linewidth} 
    l                 
    S[table-format=6.3]                                        
    S[table-format=3.2]                                        
    S[table-format=2.3]                                        
    S[table-format=1.3,round-mode=places,round-precision=3]    
    S[table-format=3.0,round-mode=places,round-precision=0]    
    S[table-format=2.0,round-mode=places,round-precision=0]    
    S[table-format=6.0,round-mode=places,round-precision=0]    
  },
  table head=\toprule
    Instrument & Fit & {$A$} & {$q$} & {$s$} & {$R^2$} & {Pieces} & {$N$} & {$L$} \\
  \midrule,
  late after last line=\\\bottomrule
]{tables/instrument_results.csv}%
 {instrument=\Instrument,
  fit=\Fit,
  A=\A,
  q=\q,
  s=\S,
  R2=\Rtwo,
  Rmax=\Rmax,
  Songs=\Songs,
  T=\TT}%
 {\Instrument & \Fit & \A & \q & \S & \Rtwo & \Songs & \Rmax & \TT}
\end{adjustbox}
\end{minipage}

\end{table}

\subsection{Experimental settings and overview}\label{subsec:data-analysis-overview}

We primarily analyze Zipf--Mandelbrot fits using \texttt{(pitch, duration)} pairs as Zipfian units, and additionally report pitch-only and duration-only fits as diagnostic comparisons.
We consider both normalized fits by fixing $A=1$ and raw fits by treating $A$ as a free parameter in equation~\ref{eq:fit}.

\begin{itemize}
  \item \textbf{Global union data.} Across all musical pieces, we considered three choices of Zipfian unit: the ordered pair of \texttt{(pitch, duration)}, pitch alone, and duration alone.
  \item \textbf{Instrument-level union data.} For each instrument class (\textit{geomungo, piri, gayageum, daegeum, haegeum, ajaeng and voice}), we constructed an instrument-specific union dataset. At this level we restricted attention to \texttt{(pitch, duration)} pairs, because the sample size within a single instrument is not sufficient to obtain stable Zipfian fits for pitch-only or duration-only units.
  \item \textbf{Piece-level data.} For each individual musical piece (the full list of all pieces are recorded in Table~\ref{tab:song-all}), we again performed Zipf--Mandelbrot fitting using \texttt{(pitch, duration)} pairs as Zipfian units.
\end{itemize}

To quantify how closely a fitted Zipf--Mandelbrot curve matches the observed rank--frequency distribution, we implemented the $R^2$-score.
Let $y_r=\log f(r)$ denote the observed log-frequency at rank $r$ and
$\hat y_r=\log \hat f(r)$ the corresponding fitted value,
for $r=1,\dots,N$. We compute
\[
R^2 \;=\; 1 \;-\;
\frac{\sum_{r=1}^N (y_r-\hat y_r)^2}{\sum_{r=1}^N (y_r-\bar y)^2},
\qquad
\bar y = \frac{1}{N}\sum_{r=1}^N y_r,
\]
which is closely related to Pearson's correlation coefficient and can be interpreted as the fraction of variance in $\{y_r\}$ explained by the fitted curve~\cite{pearson1896vii}.
We use $R^2$ primarily as a descriptive similarity index for comparing many fits across units and scopes, consistent with prior Zipf/Zipf--Mandelbrot studies in music and language~\cite{liu2013statistical,marques2012recognizing}.
While there is no universal threshold for what constitutes a ``good'' fit, values above $0.7$ are often treated as indicative of good agreement, with $0.8$ and $0.9$ reflecting stronger agreement~\cite{manaris2002progress,manaris2003evolutionary,linders2023zipf}.
The main findings are as follows:
\begin{itemize}
    \item \textbf{Global union.} Strong Zipf--Mandelbrot agreement for \texttt{(pitch, duration)} units with $R^2$-score $0.9980$ for a raw fit and $0.9979$ for a normalized fit(Fig.~\ref{fig:union_zm_and_piecewise}(a),(b), Table~\ref{tab:zm_fit_summary}\,(Panel~A)), while pitch-only is driven to parameter bounds and duration-only does not converge (Fig.~\ref{fig:union_zm_and_piecewise}(c),(d)).
    \item \textbf{Instrument-level unions.} Fits remain strong with average $R^2$-score $0.982$ for both raw and normalized fits but exhibit instrument-dependent tail behavior (Fig.~\ref{fig:inst_zm_paired}, Table~\ref{tab:zm_fit_summary}\,(Panel~B)).
    \item \textbf{Piece-level.} Most pieces admit stable fits with average $R^2$-score $0.978$ for a normalized fit(Table~\ref{tab:song-all}).
\end{itemize}

The detailed results are discussed in Sections~\ref{subsec:union-level} -~\ref{subsec:piece-level}.

\paragraph{How to read Tables~\ref{tab:zm_fit_summary}--\ref{tab:song-all}.}
Across these tables, we use a shared notation.
``Raw'' and ``norm'' denote, respectively, the raw Zipf--Mandelbrot fit and the normalized fit with $A=1$.
The fitted parameters are $A, q, s$, and $R^2$ is the $R^2$-score described in Section~\ref{subsec:data-analysis-overview}.
Here $N$ denotes the number of distinct Zipfian units, i.e. the maximum rank, and $L$ denotes the total number of Zipfian units admitting multiplicities in the corresponding dataset.

\subsection{Global union-level fits}\label{subsec:union-level}

For the global union data, the Zipf--Mandelbrot model fits the rank--frequency distribution remarkably well when \texttt{(pitch, duration)} pairs are used as Zipfian units.
Detailed parameter estimates are reported in Table~\ref{tab:zm_fit_summary} (Panel~A), and the corresponding log--log plots are shown in Fig.~\ref{fig:union_zm_and_piecewise}(a)--(b).

Figures are plotted in log--log scale. In each graph, the shaded green band marks the range of ranks where the slope of the fitted Zipf--Mandelbrot curve lies between $-1.2$ and $-0.8$. The corresponding lower and upper rank indices are reported as $\mathrm{bar}_{\min}$ and $\mathrm{bar}_{\max}$ in the table. The existence of such a nearly linear region with slope close to $-1$ indicates that the fitted model successfully captures a generalized Zipf's law.

When pitches or durations alone are used as Zipfian units, the model either fails to converge (duration) or yields a formally high $R^2$ but with parameters driven to the bounds of the optimization (pitch), indicating that the fit is not meaningful.
When \texttt{(pitch, duration)} pairs are used as Zipfian units, both the normalized and raw fits achieve an $R^2$ of about $0.998$. Here, the normalized fit fixes $A=1$ and divides frequencies by the total node count, whereas the raw fit treats $A$ as a free parameter.

\begin{figure}[t]
  \centering

  \subfloat[Union-level: normalized fit]{
    \includegraphics[width=.35\linewidth]{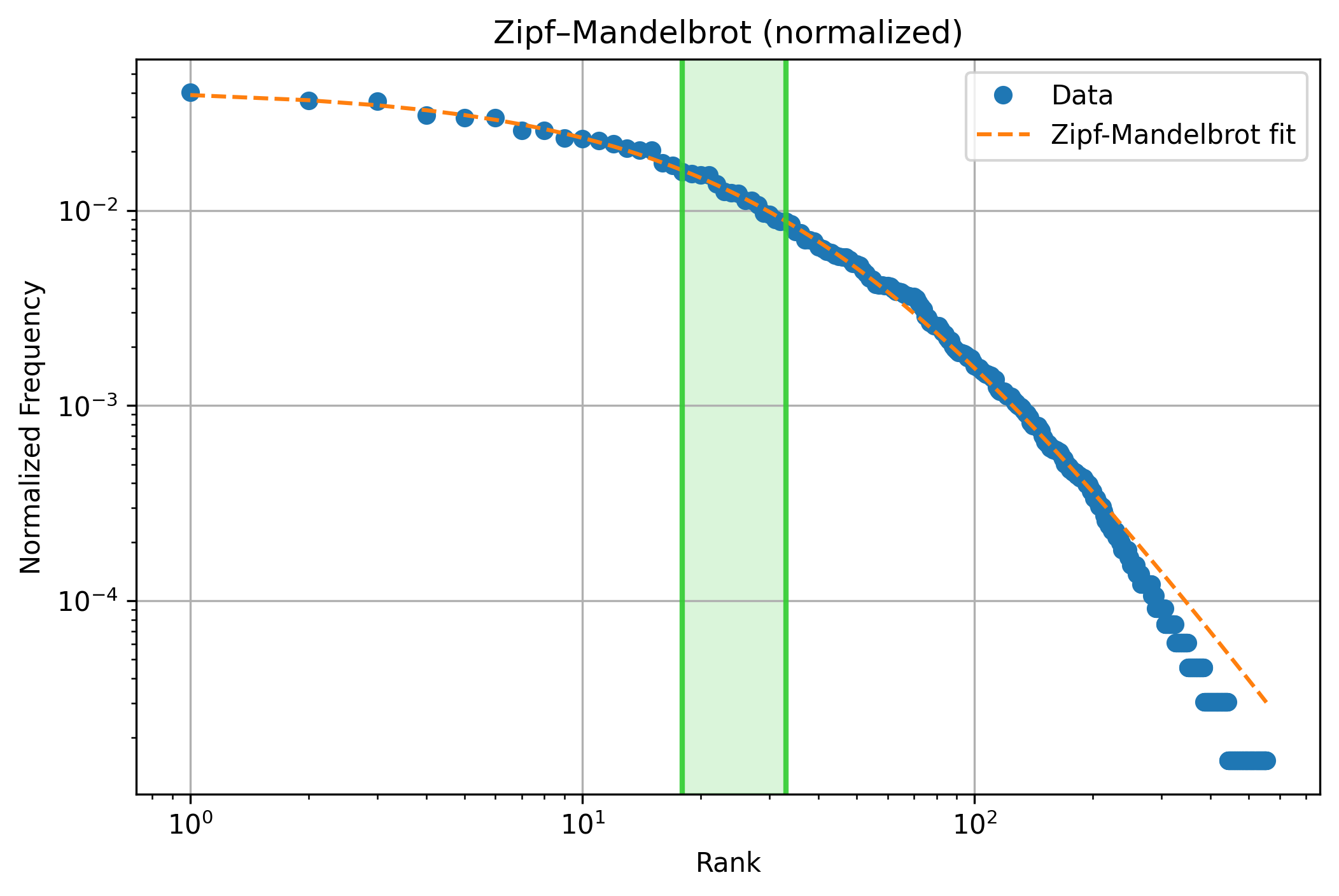}}
  \subfloat[Union-level: raw fit]{
    \includegraphics[width=.35\linewidth]{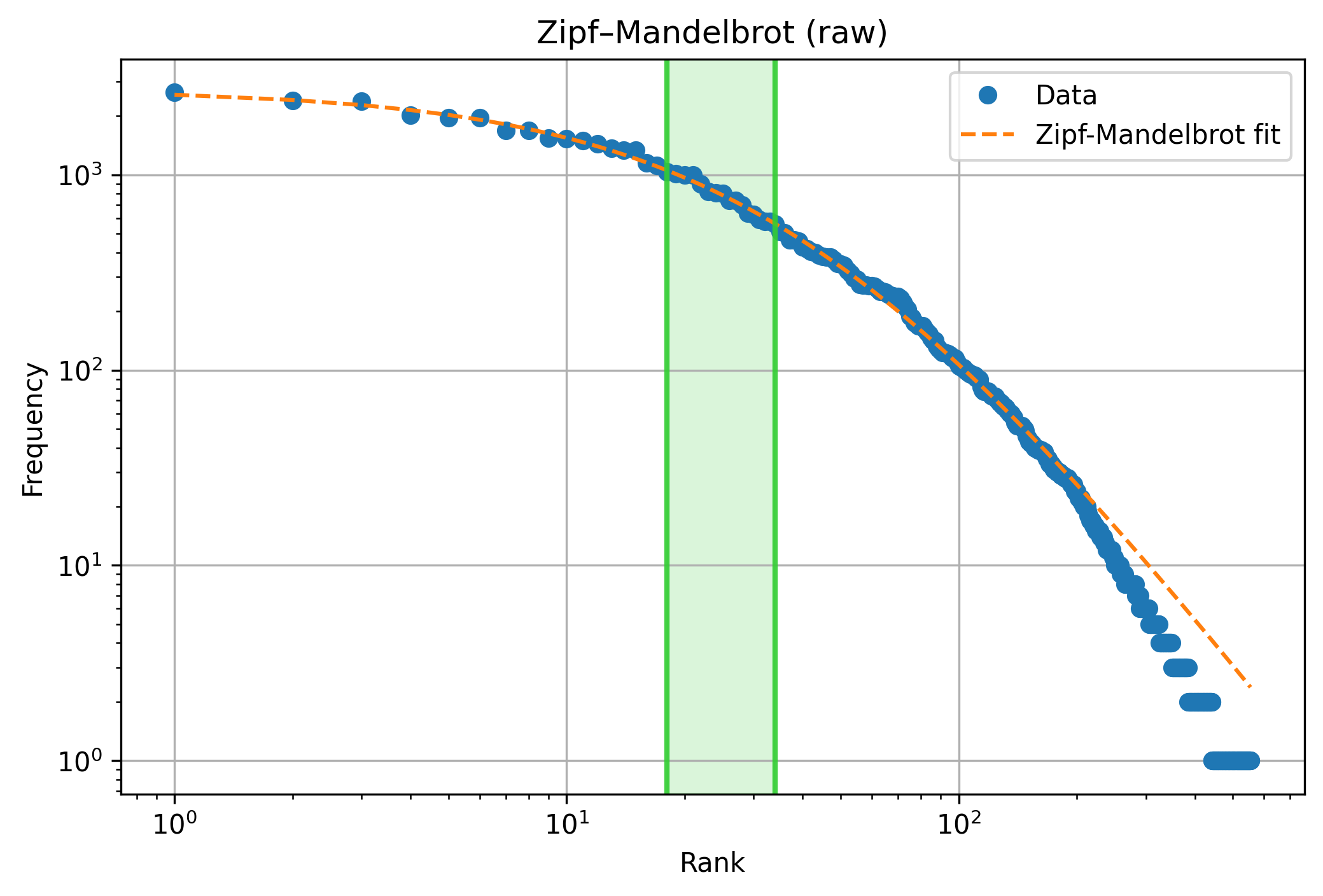}}

  \par\smallskip

  \subfloat[Pitch-only: piecewise linear fit (3 segments)]{
    \includegraphics[width=.35\linewidth]{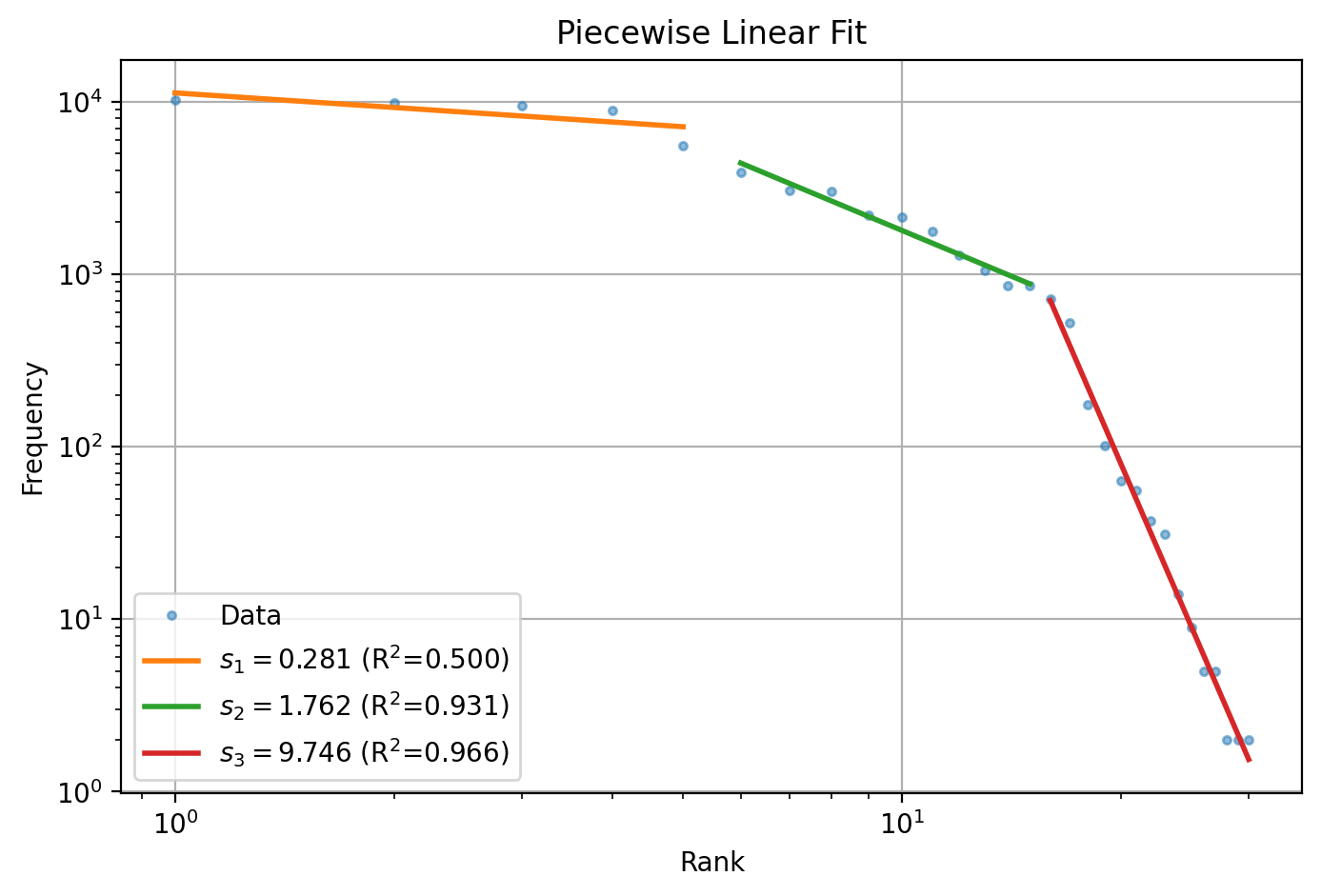}}
  \subfloat[Duration-only: piecewise linear fit (3 segments)]{
    \includegraphics[width=.35\linewidth]{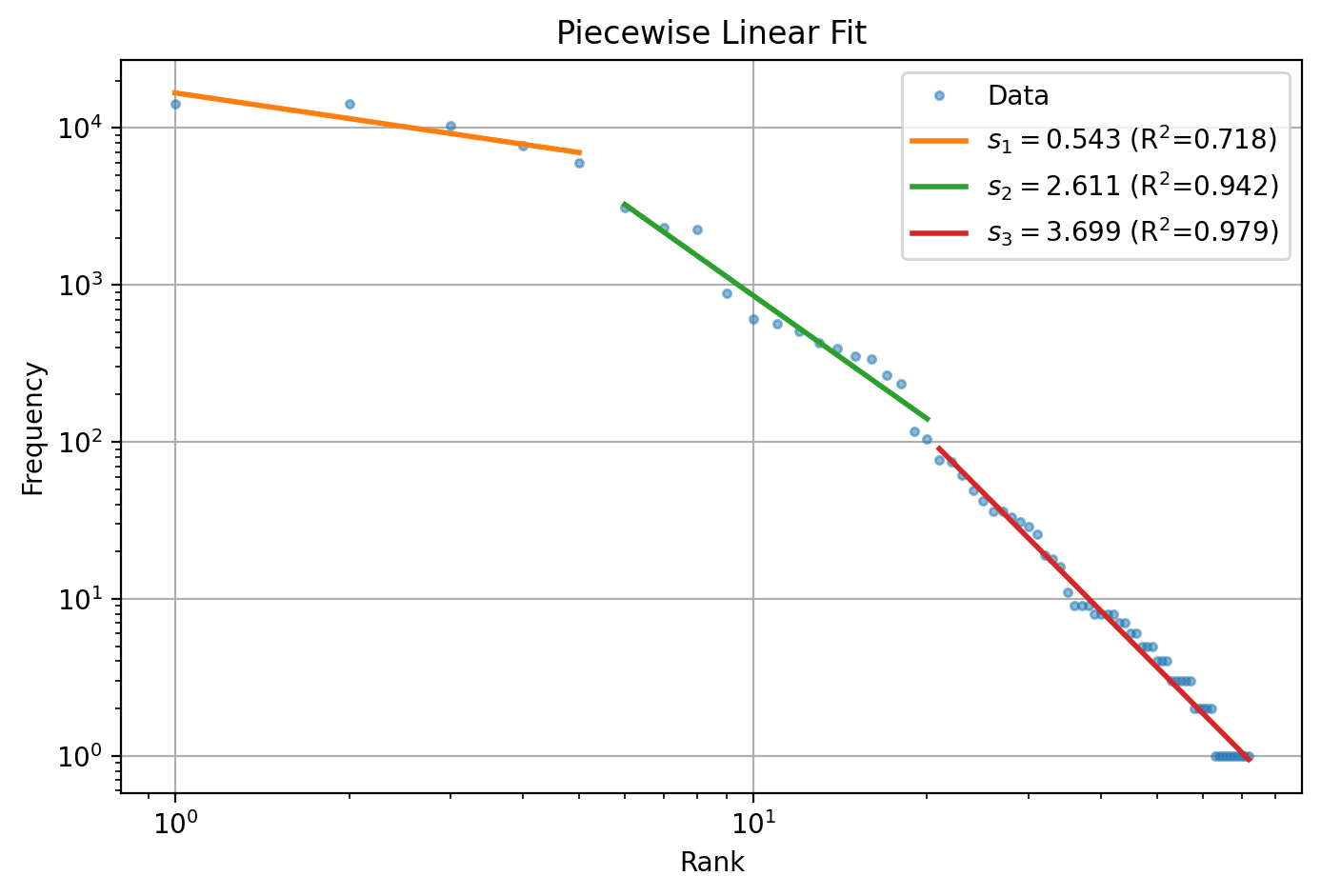}}

  \caption{Union-level Zipf--Mandelbrot fits for \texttt{(pitch, duration)} units and piecewise-linear fits for marginal pitch-only and duration-only rank--frequency curves (log--log scale).}
  \label{fig:union_zm_and_piecewise}
\end{figure}

As we pick \emph{duration} to be the Zipfian unit, the model fails to converge and fails to find the optimal solution.
On the other hand, if we pick \emph{pitch} as the Zipfian unit, the distribution of pitches has a flat head which makes hard to follow Zipf-Mandelbrot law, resulting in $s$ capping to its upper bound.
Although $R^2$ score is high, we deduced that this is not a meaningful fit.
We discuss this in detail in Section~\ref{subsec:flat_head}.

As we consider the \texttt{(pitch, duration)} pair as the Zipfian unit, we find that the model fits to the data well in overall, but it fails to fit the tail successfully, actual data presenting the weaker tail distribution. An interesting fact has been discovered as we fit the data by individual instruments.

\subsection{Instrument-level fits}\label{subsec:instrument-level}

Motivated by the union-level results in Fig.~\ref{fig:union_zm_and_piecewise}(a)--(b), we next stratify the data by instrument.
In the union fit, the agreement is strong in the mid-rank region, but the far tail shows a deviation in which the observed frequencies fall below the fitted Zipf--Mandelbrot curve, i.e., $f(r) < \hat f(r)$ for large ranks $r$.
This suggests that the union distribution may be a mixture of instrument-dependent rank--frequency regimes rather than a single homogeneous mechanism.

For each instrument class, we form an instrument-level union dataset and fit the Zipf--Mandelbrot model using \texttt{(pitch, duration)} pairs as Zipfian units.
Figure~\ref{fig:inst_zm_paired} displays the normalized and raw fits side by side, and Table~\ref{tab:zm_fit_summary}\,(Panel~B) summarizes fitted parameters and sample sizes.
The number of pieces and total observations vary substantially by instrument: for example, \textit{daegeum} contributes the largest volume ($L=19{,}985$ across $46$ pieces), followed by \textit{piri} ($L=14{,}012$, $46$ pieces) and \textit{haegeum} ($L=12{,}972$, $46$ pieces), whereas \textit{voice} is represented by only $3$ pieces ($L=904$), as each musical piece is played with different family of instruments.

Overall, the fits remain strong across instruments, with $R^2$ typically in the range $0.962$ -- $0.997$ (Table~\ref{tab:zm_fit_summary}, Panel~B), indicating that the Zipf--Mandelbrot curve captures the main rank--frequency trend.
However, the tail behavior is clearly instrument-dependent.
Most instruments exhibit close agreement between data and fit throughout the shaded mid-slope region and into the tail.
In contrast, \textit{daegeum} and \textit{ajaeng} display tail deficits: at high ranks, the fitted curve overestimates the empirical frequencies ($f(r) < \hat f(r)$), meaning that rare \texttt{(pitch, duration)} units occur less often than the fitted Zipf--Mandelbrot model predicts.
This same qualitative pattern is visible in the global union fit, suggesting that instrument-level heterogeneity is a plausible source of the union-level tail mismatch.

For \textit{ajaeng}, the overall agreement is comparatively weaker ($R^2 \approx 0.962$), and the empirical curve exhibits a noticeably more flat mid-rank region, where many \texttt{(pitch, duration)} types occur with very similar frequencies.
In musical terms, this corresponds to a distribution that spreads usage more evenly across units, so individual melodic units are less salient and the line tends to function more as a background texture.

For \textit{voice}, the lower $R^2$ ($0.943$) should also be interpreted with caution, as it is based on a small number of pieces and a much smaller sample size.
Nevertheless, the sharp drop immediately after the top few ranks in the head is consistent with the tendency for \textit{voice} parts to rely heavily on a small set of skeletal tones.

Taken together, the instrument-level fits support the interpretation that the global union curve aggregates multiple instrument-specific regimes.
In particular, deviations from the Zipf--Mandelbrot tail can arise in different ways across instruments (mild but high-volume deviations for \textit{daegeum}, and more structural deviations for \textit{ajaeng}).

\begin{figure}[t]
  \centering

  \subfloat[Geomungo (norm)]{
    \includegraphics[width=.23\linewidth]{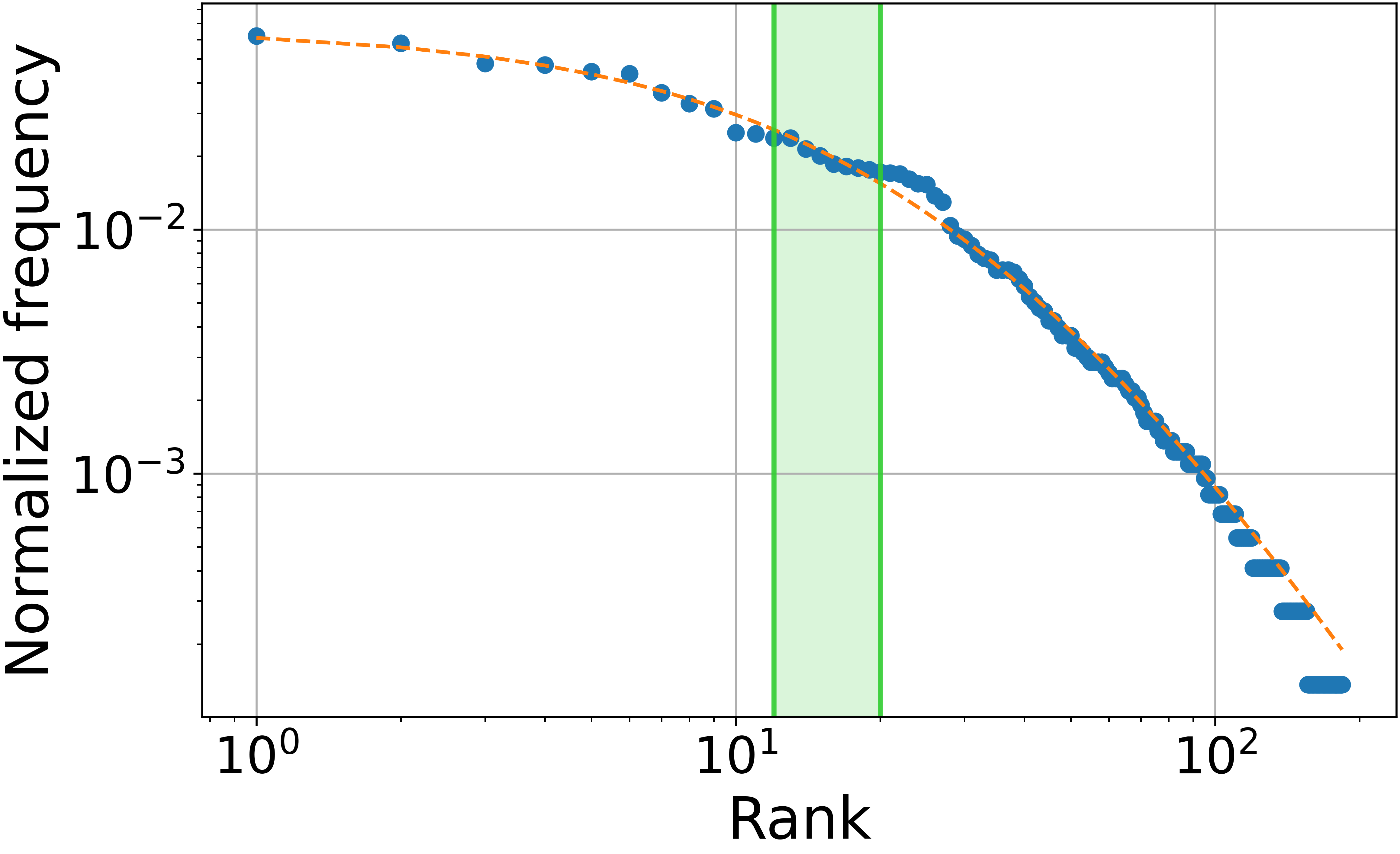}}\hfill
  \subfloat[Geomungo (raw)]{
    \includegraphics[width=.23\linewidth]{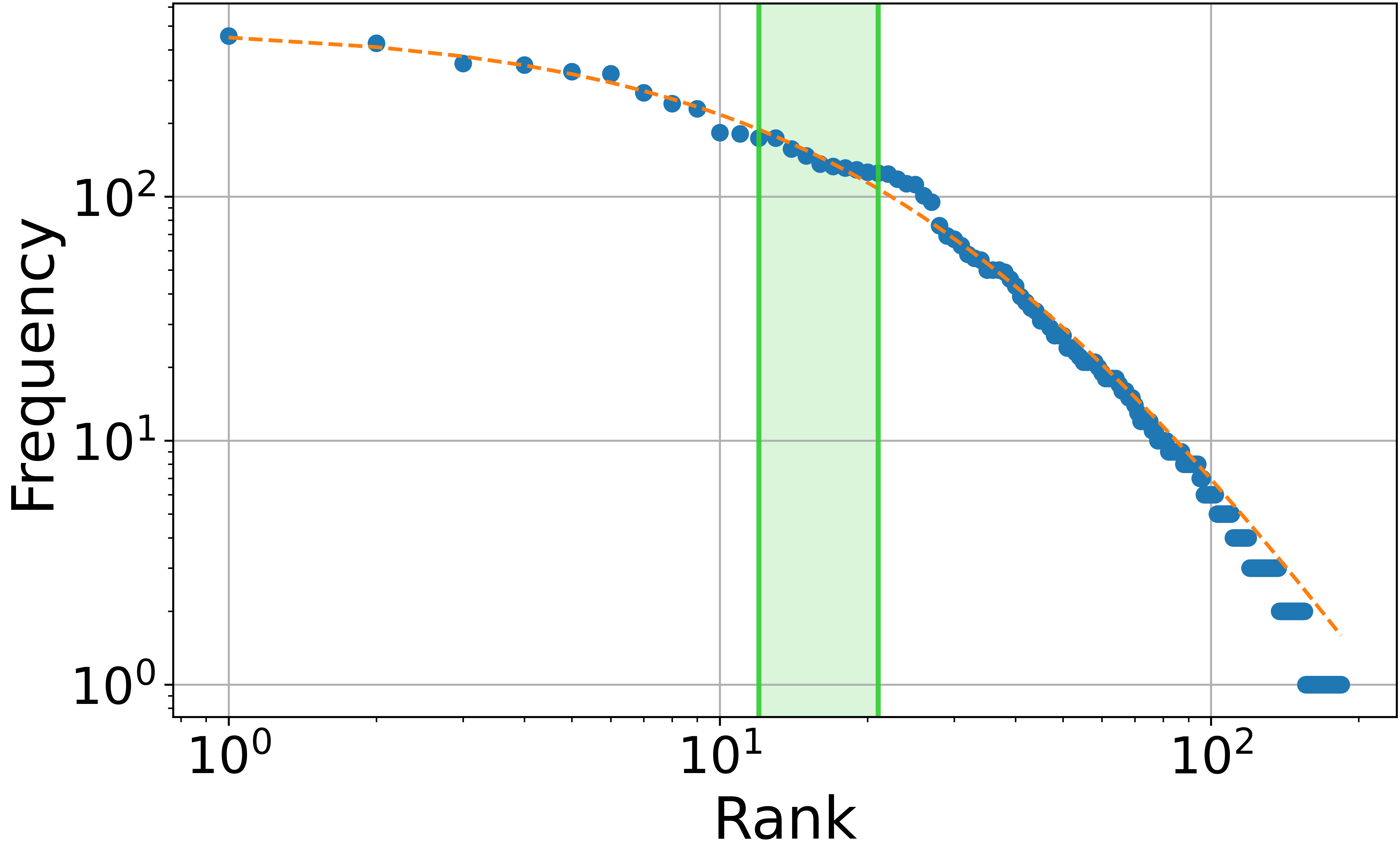}}\hfill
  \subfloat[Piri (norm)]{
    \includegraphics[width=.23\linewidth]{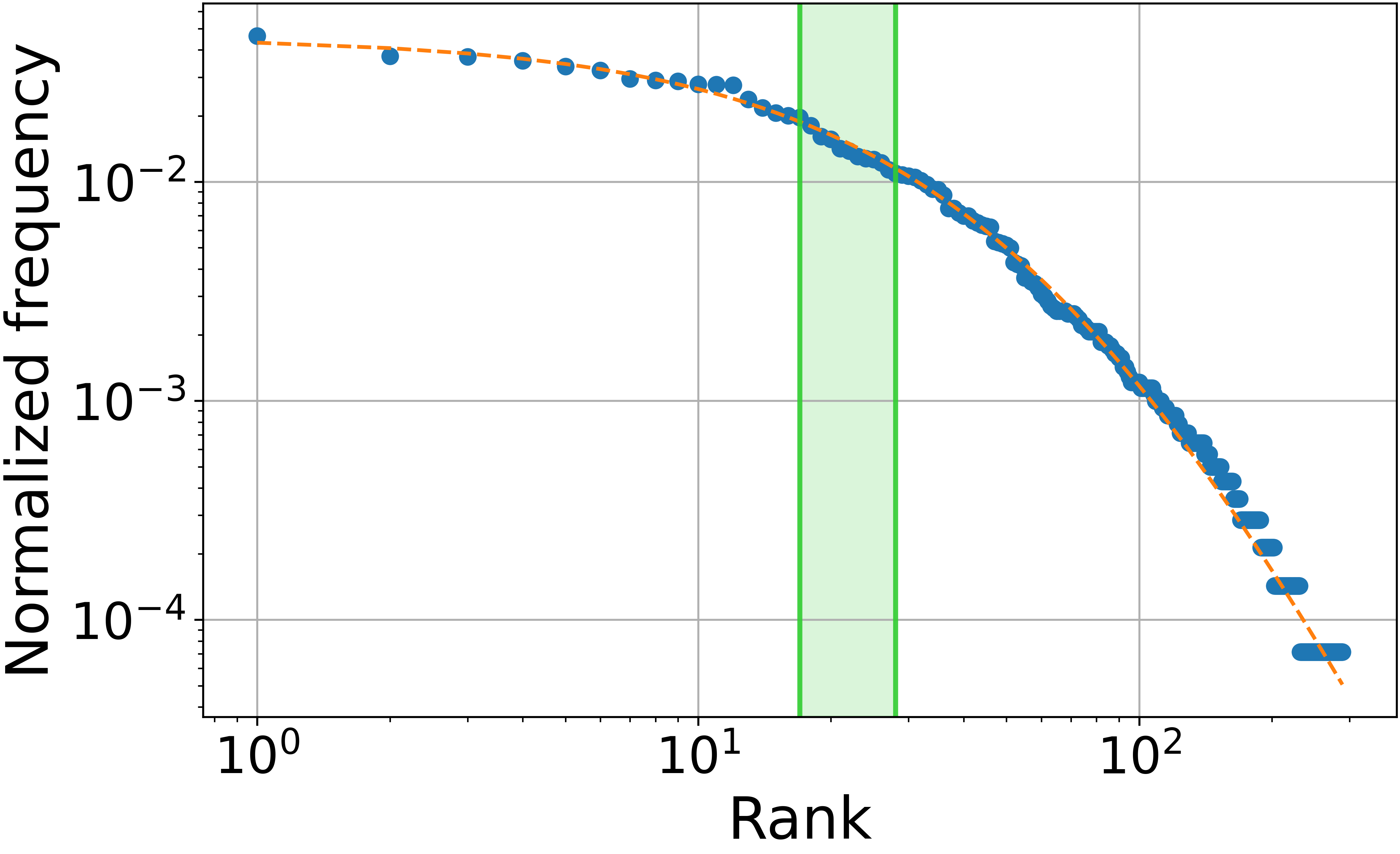}}\hfill
  \subfloat[Piri (raw)]{
    \includegraphics[width=.23\linewidth]{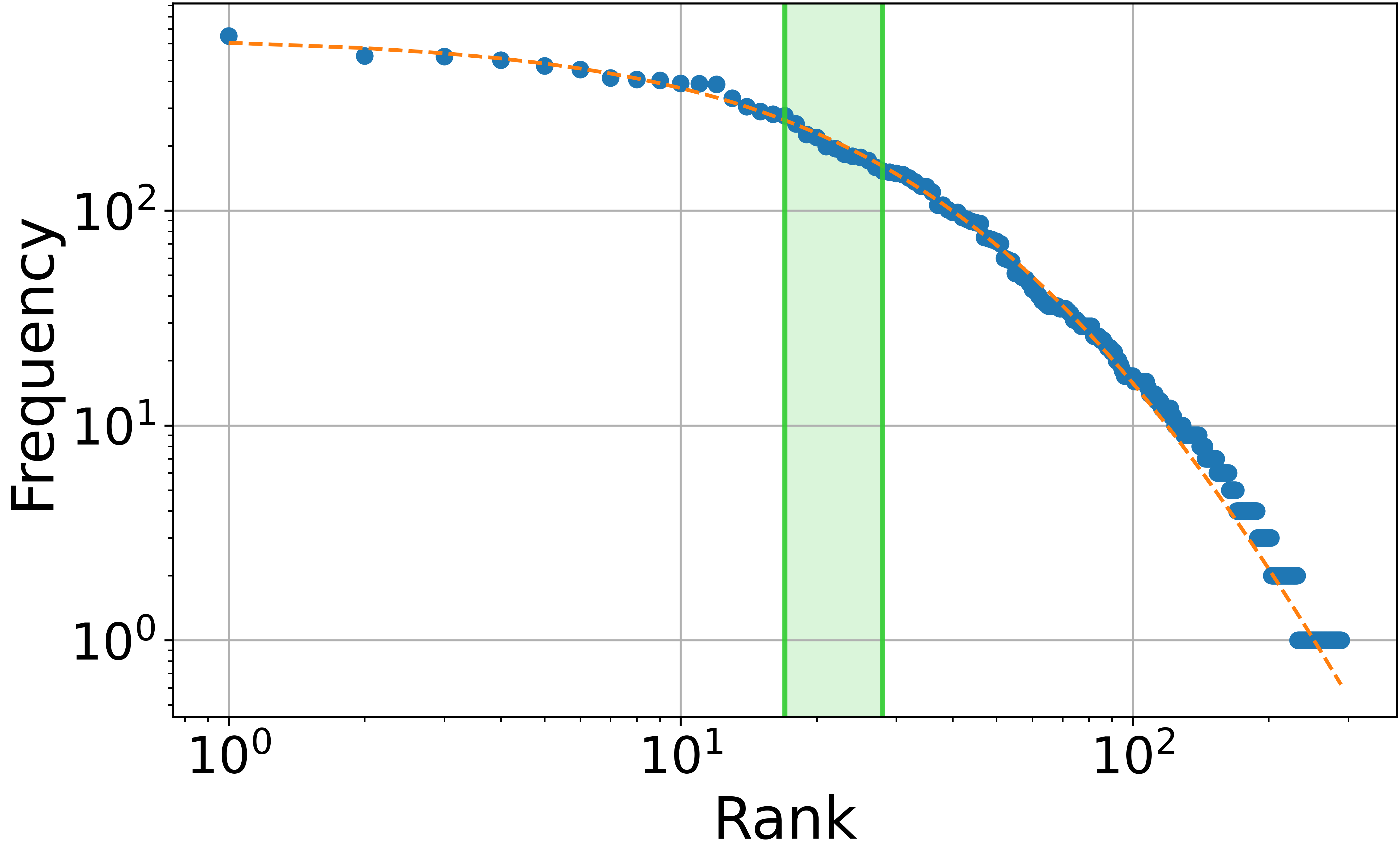}}

  \par\smallskip

  \subfloat[Gayageum (norm)]{
    \includegraphics[width=.23\linewidth]{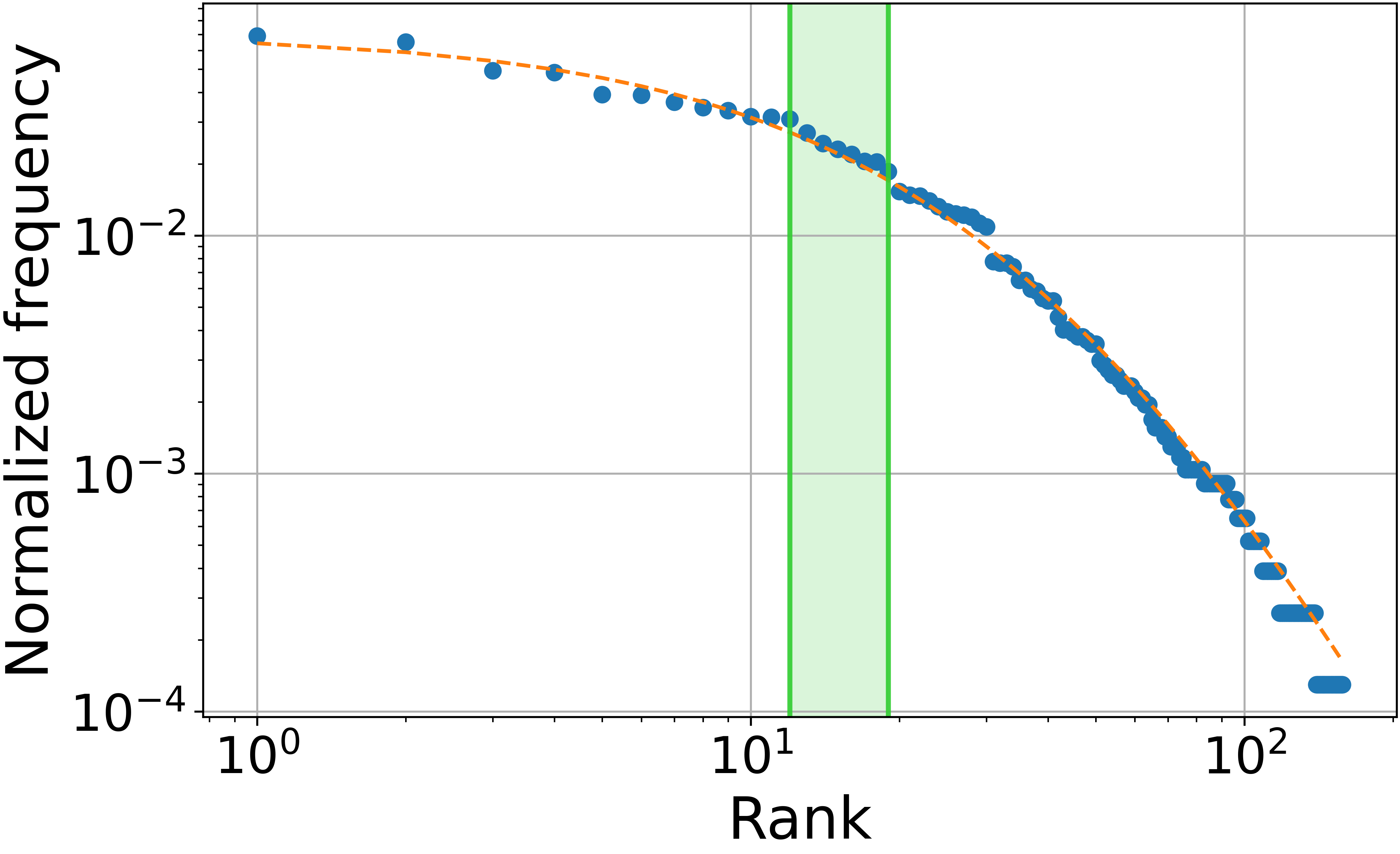}}\hfill
  \subfloat[Gayageum (raw)]{
    \includegraphics[width=.23\linewidth]{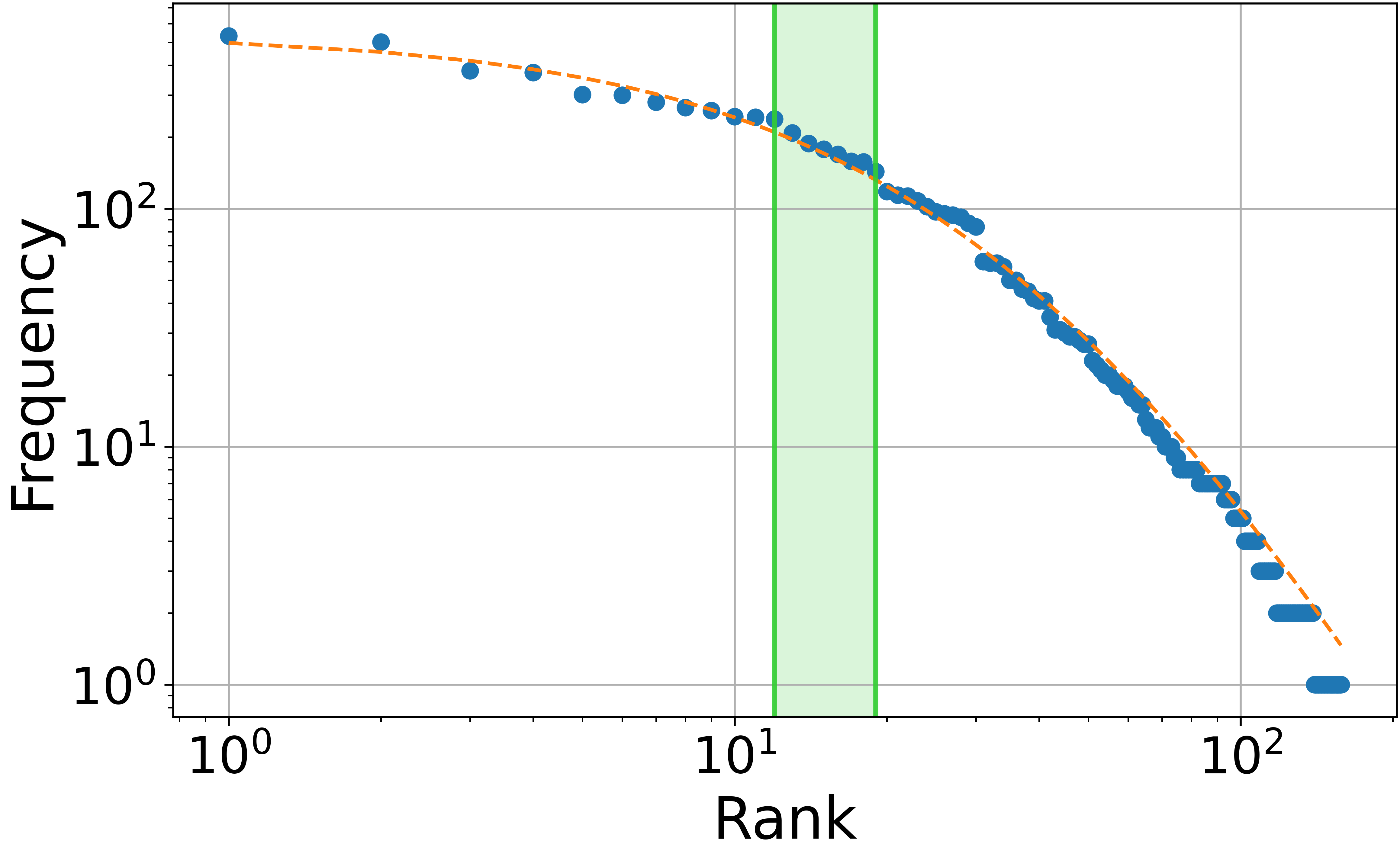}}\hfill
  \subfloat[Voice (norm)]{
    \includegraphics[width=.23\linewidth]{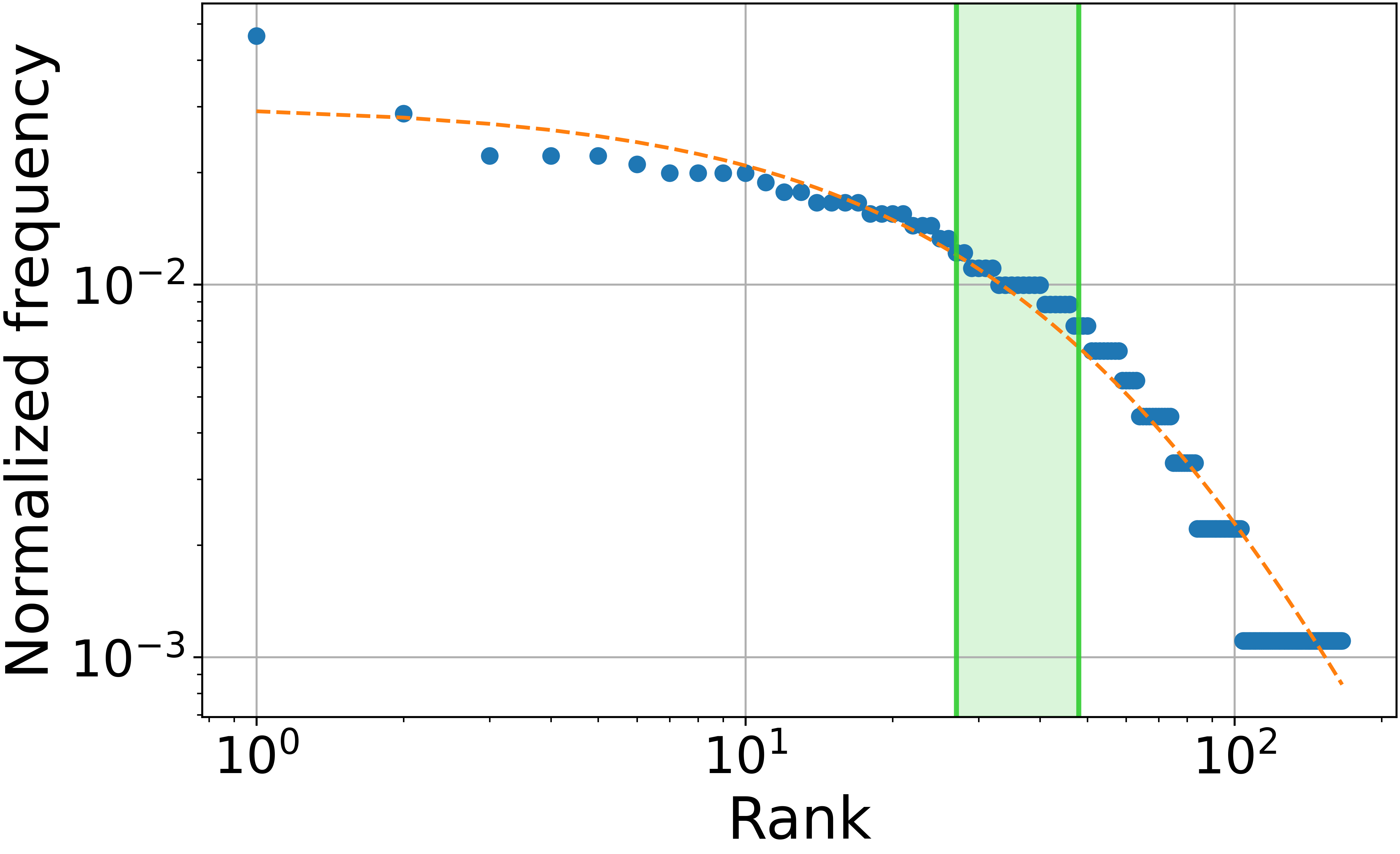}}\hfill
  \subfloat[Voice (raw)]{
    \includegraphics[width=.23\linewidth]{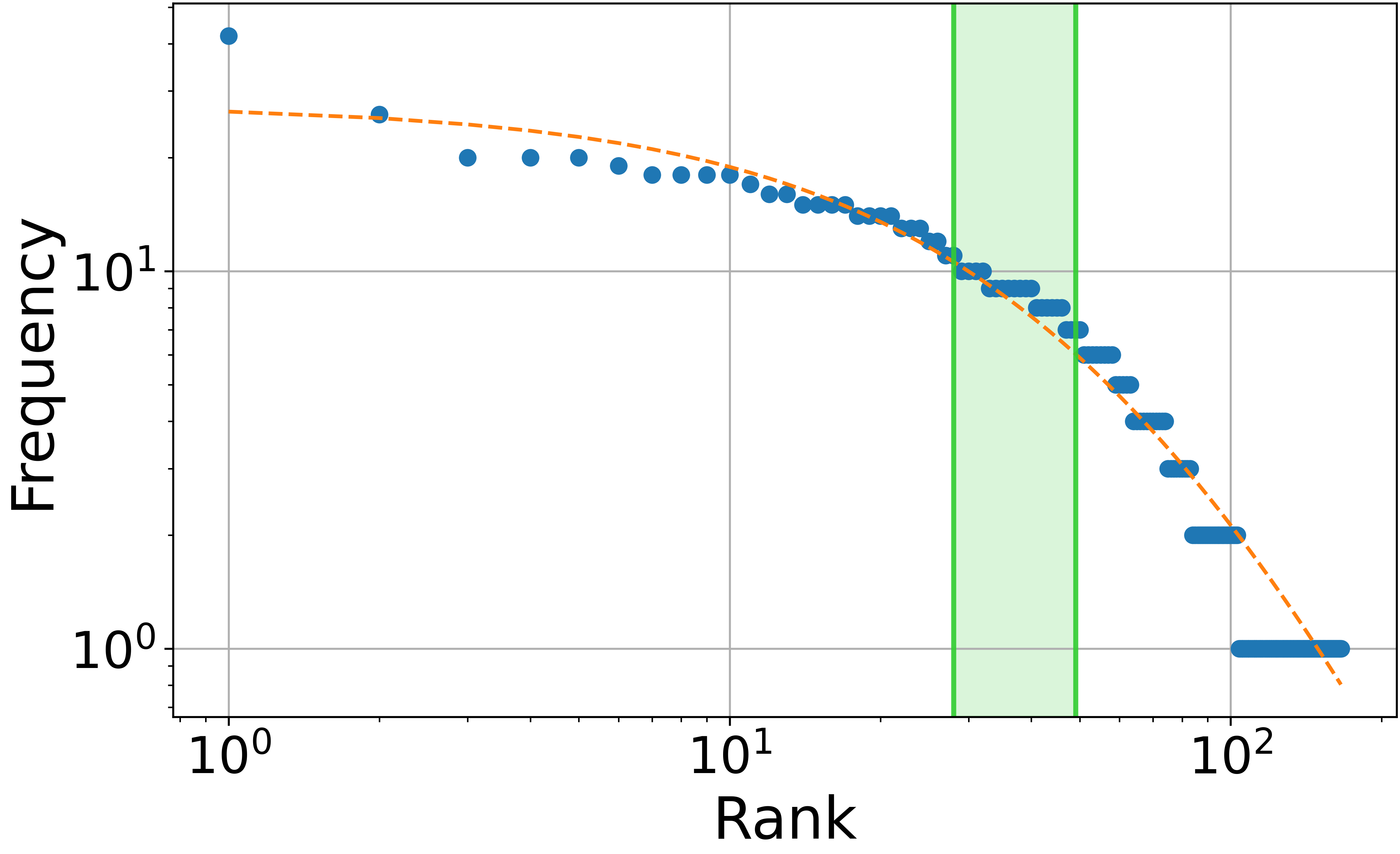}}

  \par\smallskip

  \subfloat[Daegeum (norm)]{
    \includegraphics[width=.23\linewidth]{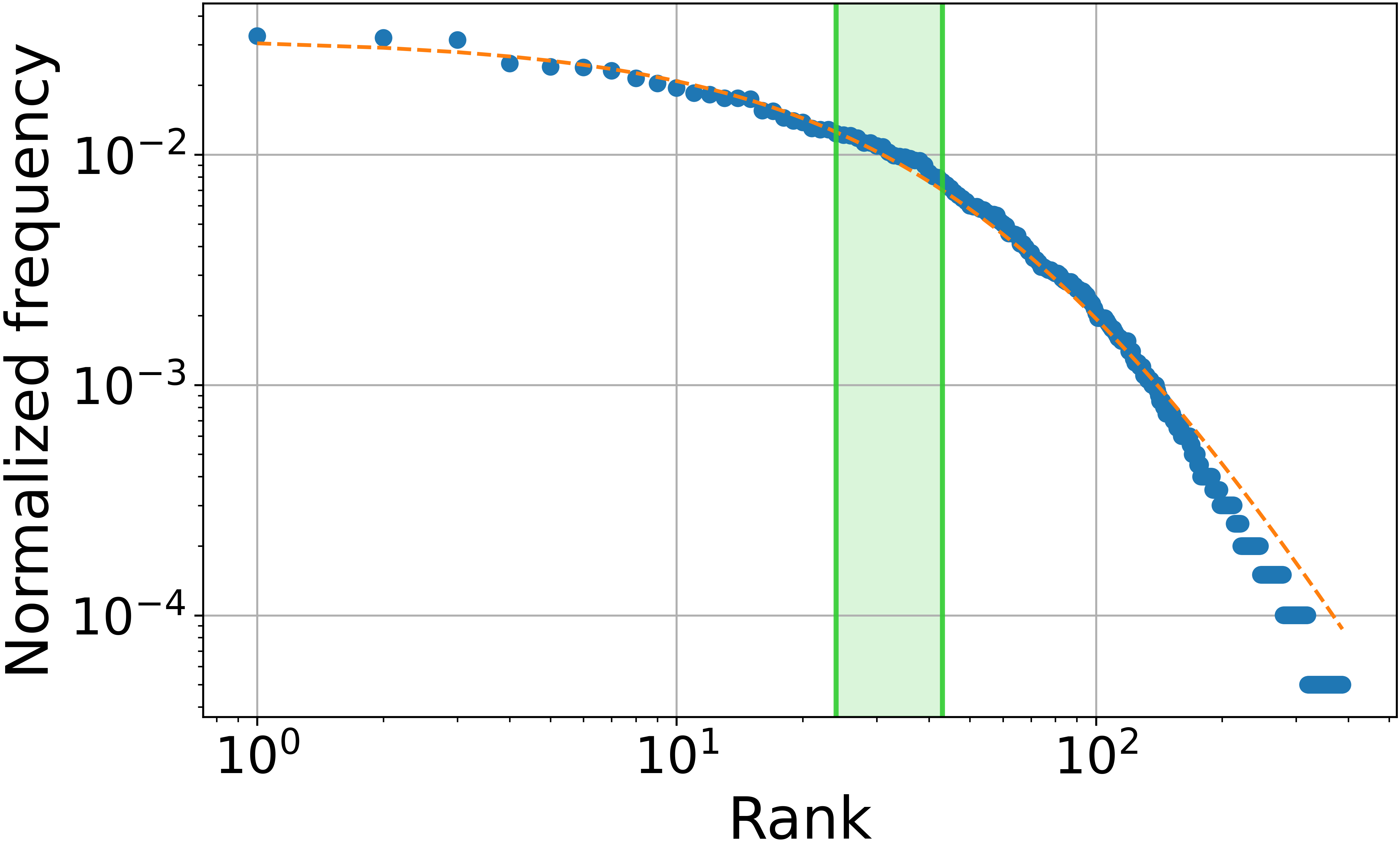}}\hfill
  \subfloat[Daegeum (raw)]{
    \includegraphics[width=.23\linewidth]{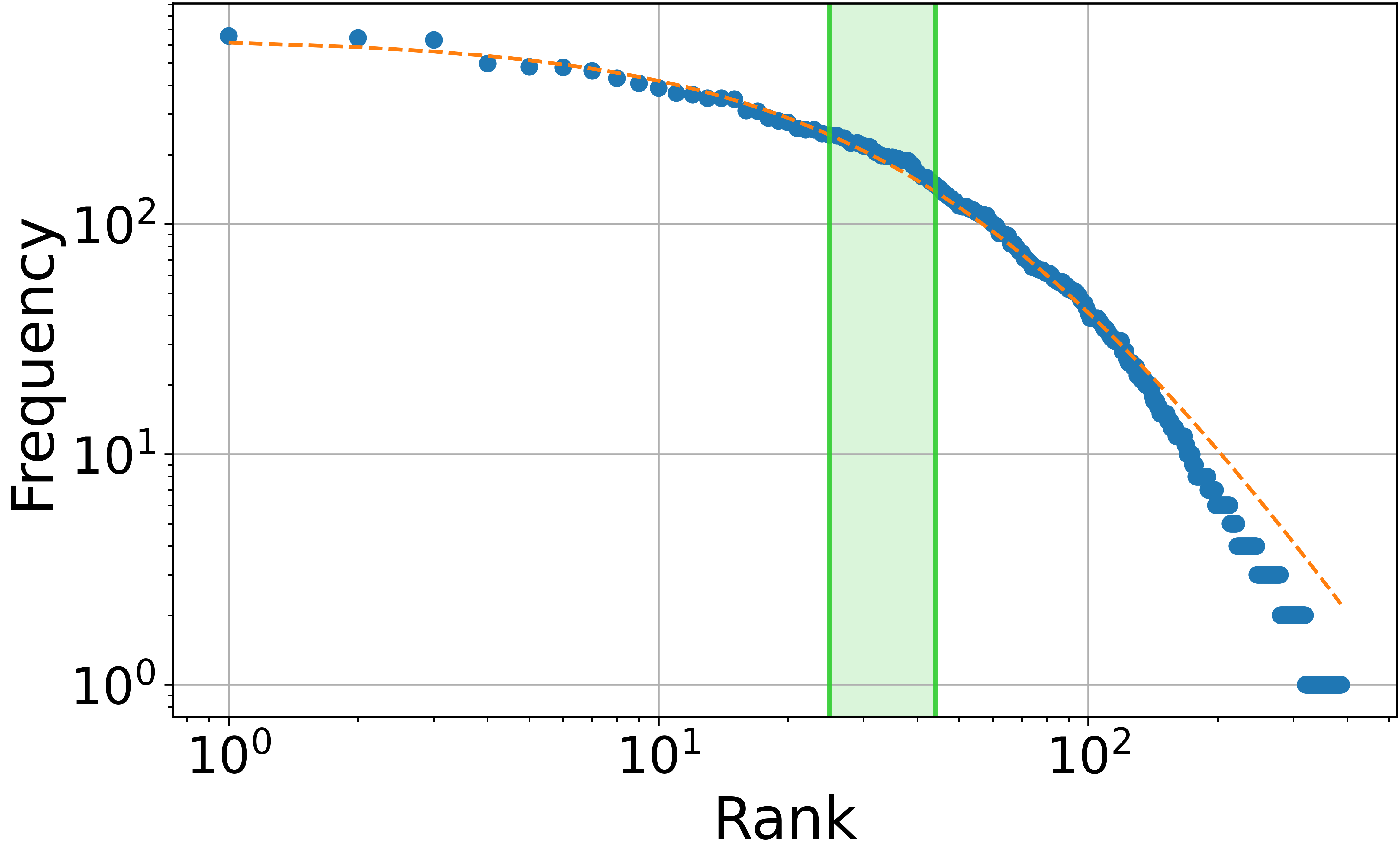}}\hfill
  \subfloat[Haegeum (norm)]{
    \includegraphics[width=.23\linewidth]{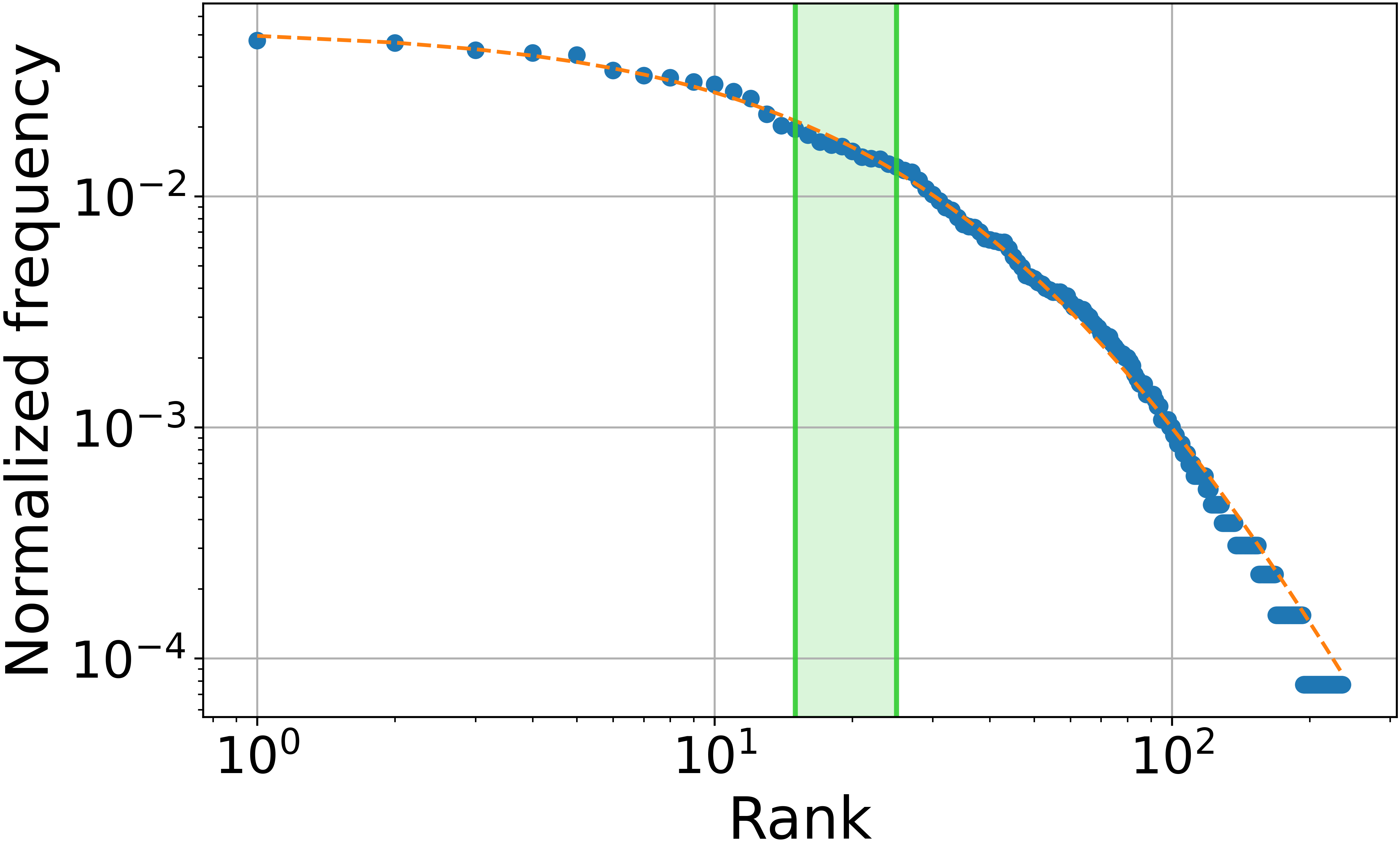}}\hfill
  \subfloat[Haegeum (raw)]{
    \includegraphics[width=.23\linewidth]{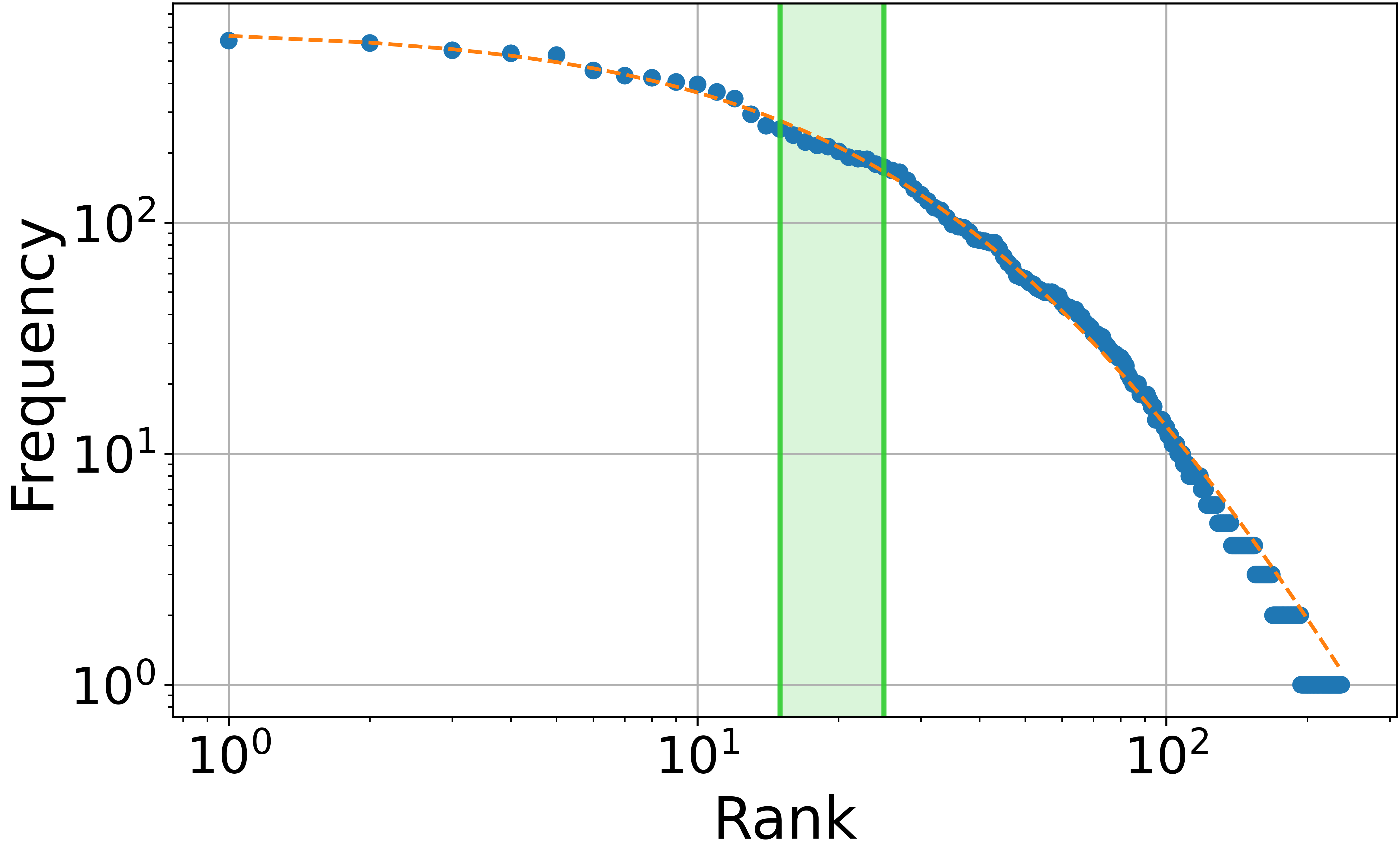}}

  \par\smallskip

  \subfloat[Ajaeng (norm)]{
    \includegraphics[width=.23\linewidth]{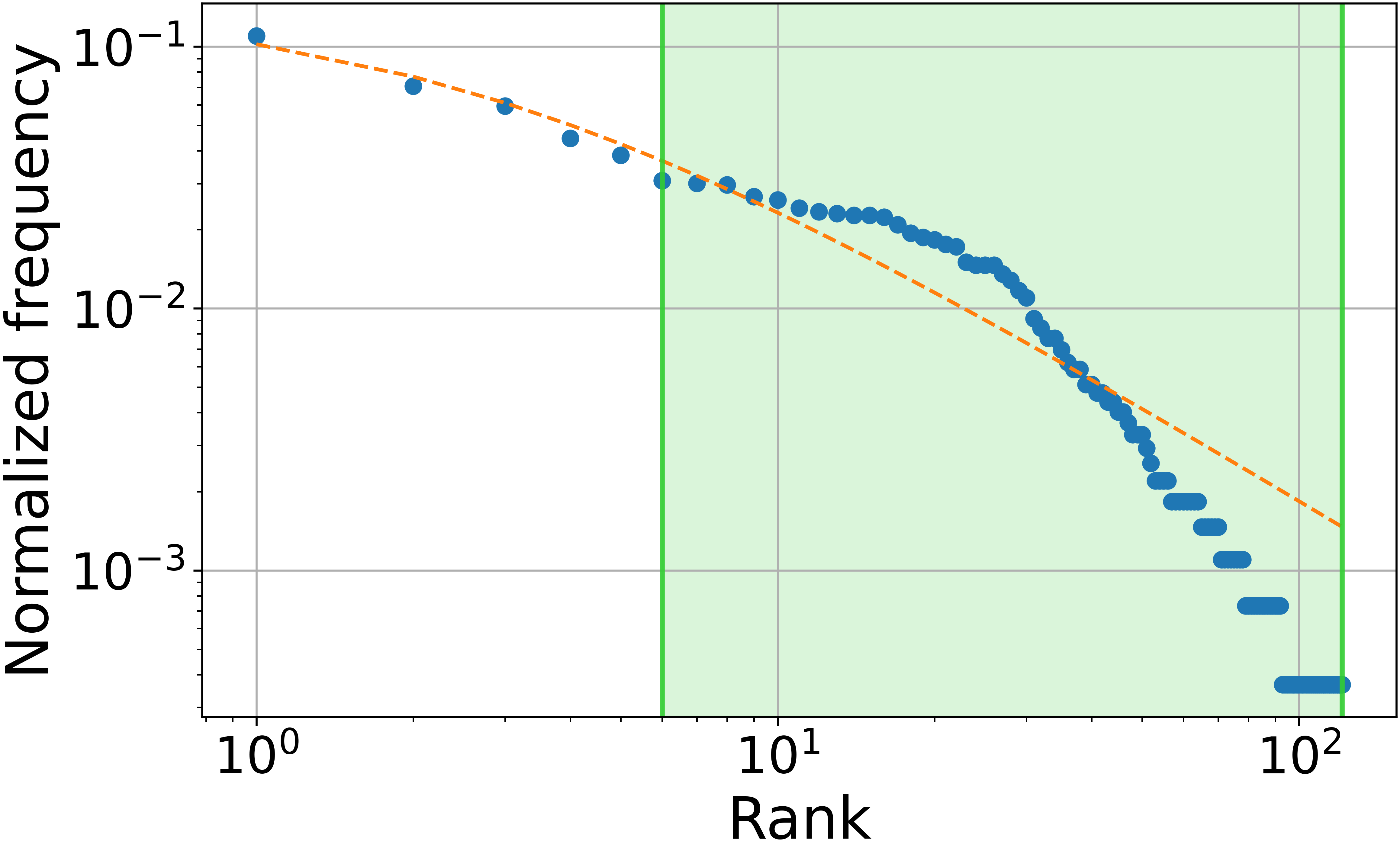}}\hfill
  \subfloat[Ajaeng (raw)]{
    \includegraphics[width=.23\linewidth]{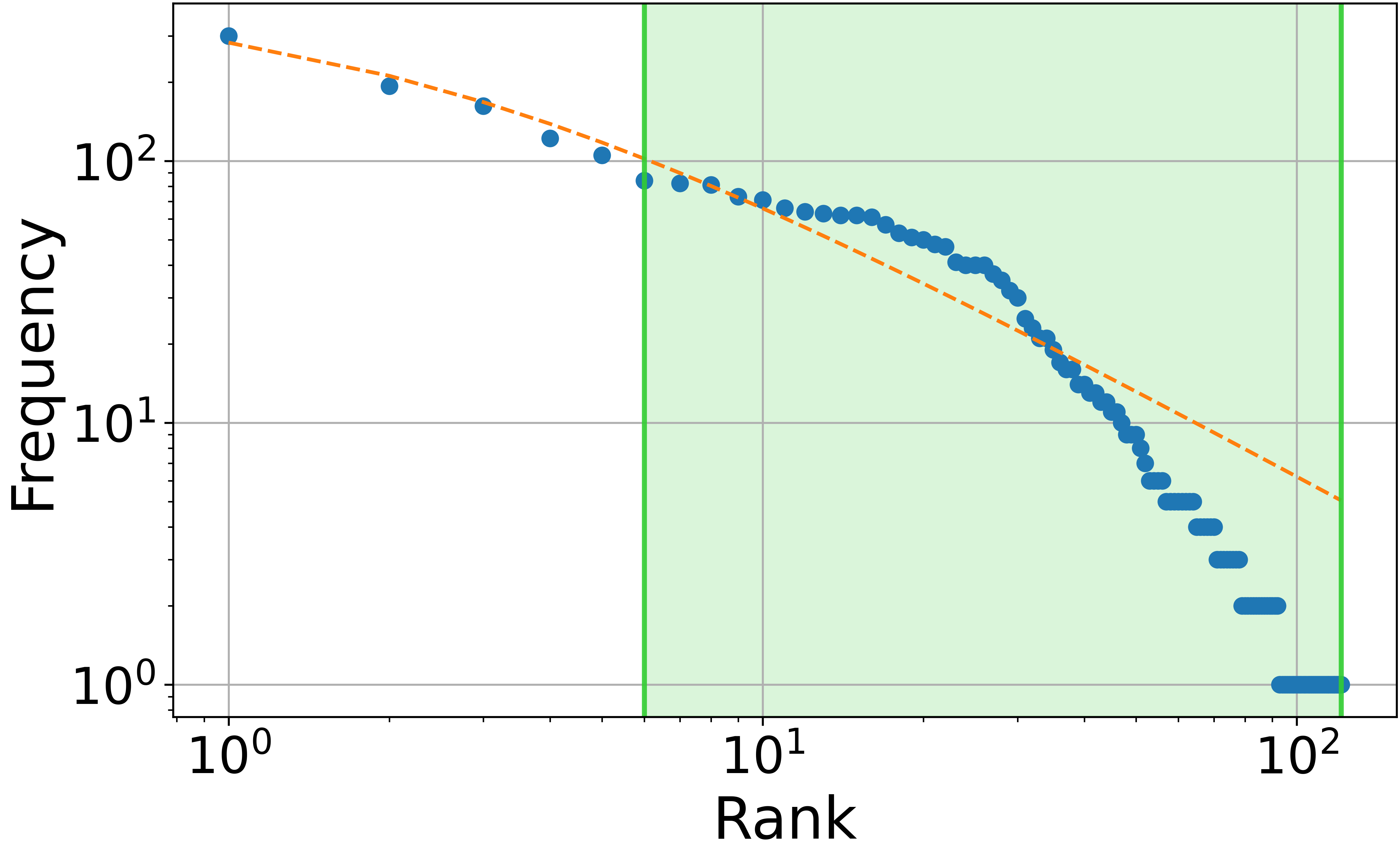}}\hfill
  \begin{minipage}[t]{.49\linewidth}
    \vspace{-50pt} 
    \centering
    \includegraphics[width=.45\linewidth]{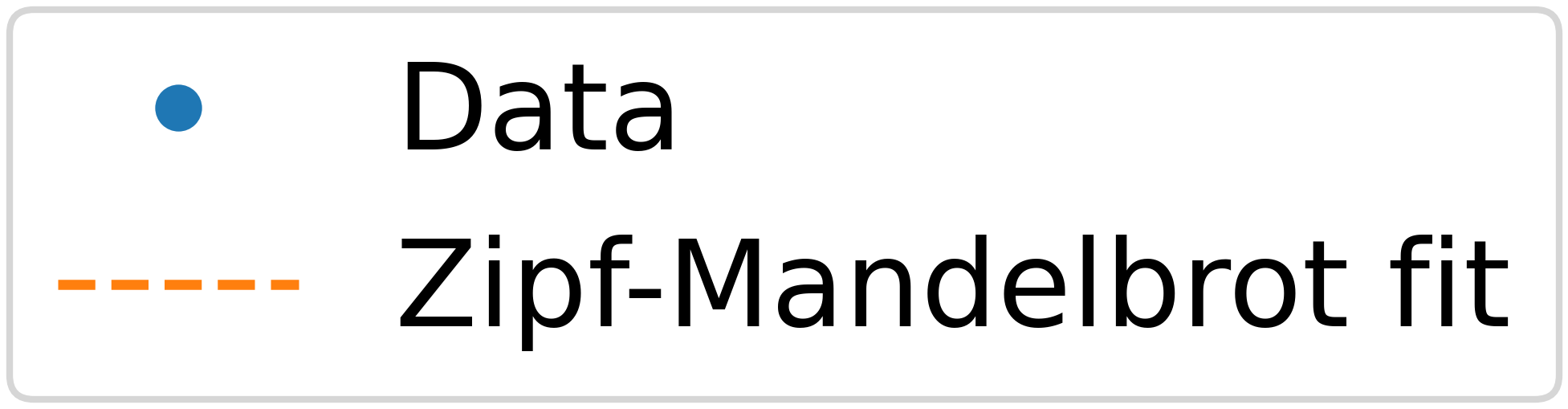}
  \end{minipage}

  \caption{Instrument-level Zipf--Mandelbrot fits for \texttt{(pitch, duration)} units (log--log scale). For each instrument, the normalized fit (left) and raw fit (right) are shown side by side.}
  \label{fig:inst_zm_paired}
\end{figure}

\subsection{Piece-level fits}\label{subsec:piece-level}

We next examine whether the Zipf--Mandelbrot regularity observed at the union level persists at the level of individual musical pieces.
Table~\ref{tab:song-all} reports normalized Zipf--Mandelbrot fits for each piece, using \texttt{(pitch, duration)} pairs as Zipfian units.

Overall, the piece-level agreement remains consistently strong.
Across the $46$ pieces in our corpus, the $R^2$ values range from $0.930$ to $0.995$ (median $0.985$), and all pieces satisfy $R^2\ge 0.9$.
Moreover, the majority of pieces exceed commonly used “strong agreement” thresholds:
$89\%$ have $R^2\ge 0.95$ and $83\%$ have $R^2\ge 0.97$.
These results indicate that the Zipf--Mandelbrot structure is not confined to the global union, but is already present at the level of individual compositions.

At the same time, the fitted parameters $(q,s)$ vary across pieces, reflecting piece-specific differences in the shape of the rank--frequency curve (e.g., how quickly the head transitions to the mid-range and how heavy the tail is).
The number of distinct \texttt{(pitch, duration)} units also varies substantially across pieces ($N$ ranges from $47$ to $213$), as does the total number of units ($L$ ranges from $380$ to $3{,}229$).
We observe that the lowest $R^2$ values tend to occur in shorter pieces.

The uniformity of strong fits across many pieces suggests that the rank--frequency organization is a robust feature of this repertoire rather than a property of the union itself.
One possible interpretation is that the corpus exhibits shared compositional feature and stylistic conventions that are repeatedly realized across pieces.

\begin{table}[t]
  \centering
  \small
  \caption{Piece-level Zipf--Mandelbrot parameters for normalized fits.
  Here ``Instruments'' denotes the number of instrument parts present in the piece.
  Superscripts $^{a}$ and $^{b}$ indicate the duration-normalization scheme used for the piece.
  The family of pieces with $^{a}$ is scaled by a factor of $\times 1.5$, where the family of pieces with $^{b}$ is scaled by a factor of $\times 0.5$.
  Parenthetical tags (G, J, P, Ch) indicate the source collection in our corpus.
  G stands for Gwanakyeongsanhoesang, J stands for Junggwangjigok, P stands for Pyeongjohoesang, and Ch stands for Chwita.}
  \label{tab:song-all}
  \begin{adjustbox}{max width=\linewidth}
    \csvreader[
      head to column names,
      separator=comma,
      tabular = {L{0.42\linewidth}
                 S[table-format=1.0, round-mode=places, round-precision=0]
                 S[table-format=4.3]
                 S[table-format=2.3]
                 S[table-format=4.0, round-mode=places, round-precision=0]
                 S[table-format=6.0, round-mode=places, round-precision=0]
                 S[table-format=1.3]},
      table head=\toprule
        Title & {Instruments} & {$q$} & {$s$} & {$N$} & {$L$} & {$R^2$} \\
        \midrule,
      late after last line=\\\bottomrule
    ]{tables/song_results_tex_sorted.csv}
    {title_tex=\Title, Instruments=\Inst, q=\Q, s=\S, Rmax=\Rmax, T=\TT, R2=\Rtwo}%
    {\Title & \Inst & \Q & \S & \Rmax & \TT & \Rtwo}
  \end{adjustbox}
\end{table}

\section{
Remarks on Head Behavior and Joint Distributions
}\label{sec:Discussion}

\subsection{Why does the flat head cause a fitting failure?}
\label{subsec:flat_head}
Consider the Zipf–Mandelbrot model \(f(r) = A(r + q)^{-s}\), and its local slope on a log-log plot:
\begin{equation}
\label{eq:slop_log}
    \frac{d\log f}{d\log r} \;=\; -\,s\,\frac{r}{r + q}.
\end{equation}
Suppose the empirical head is (nearly) flat over the range \(1 \le r \le r_h\), which can be expressed as
$ \displaystyle \Bigl|\tfrac{d\log f}{d\log r}\Bigr| = s\,\frac{r}{r + q} \le \epsilon. $
Then,
    $ \displaystyle q \;\ge\; \Bigl(\tfrac{s}{\epsilon} - 1\Bigr)\,r, $
for \(1 \le r \le r_h\).

Thus, to maintain a sufficiently flat head across a wide observed range (large \(r_h\)), the shift parameter \(q\) must be large. Geometrically, this shifts the bend point (which occurs around \(r \approx q\)) to the right. According to Eq.~\ref{eq:slop_log}, when \(r \ll q\), the slope approaches \(0\), while when \(r \gg q\), it approaches \(-s\). Therefore, a large \(q\) ensures that the observed region (the portion of the data where observations are available, typically \(1 \le r \le N\)) remains in the \(r \ll q\) regime.

Because real datasets have a finite range \(1 \le r \le N\), taking \(q\) to be very large keeps \(\tfrac{r}{r + q}\) small throughout the observed region, removing the bend point. 
This makes it difficult to identify the transition to the tail slope \(-s\) reliably. 

However, if a bend point exists within the observed range (\(1 \le r \le N\)), the fitted parameter \(s\) tends to increase simultaneously to compensate, since the slope \(\frac{d\log f}{d\log r} = -s\,\frac{r}{r + q}\) must still fit the data curvature. In practice, Table~\ref{tab:zm_fit_summary} (Pannel A) illustrates this behavior when the pitch is used as the Zipfian unit: the fitted value of \(q\) is \(96.90\), while \(N = 30\), implying that the entire observed range lies within the \(r \ll q\) regime. Consequently, the estimated \(s\) reaches its upper bound (\(s = 20\)).

\subsection{Joint distribution of two independent Zipfian units}\label{subsec:joint}

We observed that the rank--frequency distribution of \emph{(pitch, duration)} units is almost identical to that of pairs obtained by independently sampling pitch and duration from their probability distributions derived from our dataset, reporting the $R^2$-score of 0.993. 
While the rank--frequency distribution of \emph{(pitch, duration)} units follows the Zipf-Mandelbrot law, we found that each rank--frequency distribution of \emph{pitch} and \emph{duration} units follows the piecewise Zipf's law with 3 line segments.
From this observation, in this section we verify if the joint distribution of two independent distributions following the Zipf law follows the Zipf-Mandelbrot model or not.
For the convenience, we simplify the piecewise Zipf model into the Zipf model. Let
\[
f_{X_1}(x_1)=x_1^{-t_1},\qquad 
f_{X_2}(x_2)=x_2^{-t_2},\qquad t_1,t_2>0, \]
\[
f_{X_i}: X_i \to \mathbb{R}_{\ge 0}, \qquad X_i = [1, \infty), \qquad i=1,2,
\]
and assume further that \(X_1\) and \(X_2\) are independent. The joint density function is
\[
f(x_1,x_2)=f_{X_1,X_2}(x_1,x_2) = x_1^{-t_1}x_2^{-t_2},\qquad (x_1,x_2)\in[1,\infty)^2.
\]
For a threshold \(f_0\in(0,1]\), define the \emph{upper level set} of $f_0$ as
\[
\Omega(f_0):=\bigl\{(x_1,x_2)\in[1,\infty)^2:~ f(x_1,x_2)\ge f_0 \bigr\},
\]
and its area \( A: (0,1] \to \mathbb{R}_{\ge 0} \) as
\[
A(f_0):=\operatorname{Area}\bigl(\Omega(f_0)\bigr) = \iint_{\Omega(f_0)} dxdy.
\]
Since \(f\) is strictly decreasing in each coordinate, \(A\) is strictly decreasing on \((0,1]\) and hence is invertible. 

Now we want to connect the function $A$ with the rank-frequency distribution of a joint distribution.
For $N_1,N_2\in\mathbb{N}$, define the lattice in $[1,\infty)^2$ by
\[
\mathcal{L}(N_1,N_2)
:=\Bigl\{\Bigl(\frac{m}{N_1},\frac{n}{N_2}\Bigr)\in[1,\infty)^2:\; m,n\in\mathbb{N}\Bigr\}.
\]
As $f$ is strictly decreasing in each coordinate, we can find the $r$-th largest element in the following set (whose existence is guaranteed in the $r \times r$ lattice):
\[
\Big\{f(x_1, x_2) \; \big| \; (x_1, x_2) \in \mathcal{L}(N_1 , N_2 )\Big\}
\]
Define \(f_{*(N_1 , N_2 )} : \frac{1}{N_1 N_2}\mathbb{N} \to \mathbb{R}_{> 0}\) as
\[\frac{r}{N_1 N_2} \longmapsto \text{the } r\text{-th largest value of } \Big\{f(x_1, x_2) \; \big| \; (x_1, x_2) \in \mathcal{L}(N_1 , N_2 )\Big\},\]
where the set $\frac{1}{N}\mathbb{N} \coloneqq \{\frac{n}{N} \;|\; n \in \mathbb{N}\}.$

\begin{remark}\label{rmk:rank-freq}
    If $N_1 = N_2 = 1$, then the function $f_{*(1,1)}$ is exactly the same as rank-frequency plot for the joint distribution.
\end{remark}

For $f_0\in(0,1]$, define the normalized discrete counting functional
\[
A_{N_1,N_2}(f_0)
:=\frac{1}{N_1N_2}\#\bigl(\mathcal{L}(N_1,N_2)\cap\Omega(f_0)\bigr).
\]

Finally, we state the connection between $f_{*(N_1, N_2)}$ and $A^{-1}$. 
The detailed proof is illustrated in Appendix~\ref{sec:prop_proof}.

\begin{proposition}[Sorting converges to the inverse-area curve]\label{prop:sorting_limit}
Fix $r>0$ and $\varepsilon_1,\varepsilon_2>0$. Then there exist $M_1,M_2\in\mathbb{N}$ such that
for every $N_1\ge M_1$ and $N_2\ge M_2$ there exists an integer $c\in\mathbb{N}$
for which
\[
\Bigl|\,r-\frac{c}{N_1N_2}\Bigr|<\varepsilon_2
\qquad\text{and}\qquad
\Bigl|\,A^{-1}(r)-f_{*(N_1,N_2)}\Bigl(\frac{c}{N_1N_2}\Bigr)\Bigr|<\varepsilon_1.
\]
In particular, the curve $\{(r,A^{-1}(r))\;|\;r\ge 0\}$ is the continuous analogue of the
rank--frequency plot obtained by sorting lattice samples by decreasing $f$.
\end{proposition}

In proposition~\ref{prop:sorting_limit}, we've shown that the rank--frequency curve of the joint distribution
$f(x_1,x_2)=x_1^{-t_1}x_2^{-t_2}$ is exactly the function $f_*(r)=A^{-1}(r)$.
To test whether this curve is well-approximated by the Zipf--Mandelbrot law,
it is necessary to compute $A(f_0)$ explicitly, which we do next. Likewise, the detailed proof is written in Appendix~\ref{sec:prop_proof}.

\begin{proposition}[Closed form of the sorting function]\label{prop:sorting}
Fix $t_1,t_2>0$. For $f_0\in(0,1]$, define the superlevel set
\[
\Omega(f_0)\;:=\;\bigl\{(x_1,x_2)\in[1,\infty)^2:\;x_1^{-t_1}x_2^{-t_2}\ge f_0\bigr\},
\]
and let $A(f_0):=\operatorname{Area}\,\Omega(f_0)$. Then:
\begin{enumerate}
\item If $t_1\neq t_2$,
\[
A(f_0)\;=\;1+\frac{t_1\,f_0^{-1/t_1}-t_2\,f_0^{-1/t_2}}{\,t_2-t_1\,}.
\]
\item If $t_1=t_2=t$,
\[
A(f_0)\;=\;f_0^{-1/t}\!\Bigl(\tfrac{1}{t}\ln\tfrac{1}{f_0}-1\Bigr)+1.
\]
\end{enumerate}
Moreover, $A:(0,1]\to[1,\infty)$ is strictly decreasing and hence bijective; the
continuous rank--frequency curve is the inverse
\[
f_*(r)\;=\;A^{-1}(r)\qquad(r\ge 1).
\]
\end{proposition}

With the explicit formula for $A(f_0)$ obtained, the rank--frequency curve $f_*(r)=A^{-1}(r)$ can be computed for any $(t_1,t_2)$.
We now verify whether this theoretically calculated curve is well-approximated by the Zipf--Mandelbrot law by fitting the ZM model to $\{(r,f_*(r))\}$ and reporting the resulting $R^2$ scores.
\paragraph{Verification.}
Given \((t_1,t_2)\), we fit the ZM model to \(\{(r,f_*(r))\}\) using \texttt{scipy.optimize.curve\_fit}
to obtain \((A,q,\alpha)\).
The verification is held for $(t_1 , t_2)$ from $0.5$ to $4$, scaled by $0.5$ each.
We further calculated the $R^2$-score for each fit. The results are recorded in Table~\ref{tab:r2_matrix_joint}.

The above result 
has several important theoretical implications. First, it could explain the stability of heavy-tailed behavior under pairing or joint construction. While Zipf’s law characterizes marginal rank-frequency distributions, many real-world systems involve interactions between two or more Zipfian units. Demonstrating that their joint distribution follows the Zipf-Mandelbrot law suggests that the Zipf-Mandelbrot law may form naturally as a higher-order extension of Zipfian scaling. The result also helps to explain the relationship between the Zipf’s law and the Zipf-Mandelbrot law. The latter is an empirical extension of Zipf’s law with an additional shift parameter. Our result may imply a principled generative way of the emergence of the Zipf-Mandelbrot law. To the best of the authors’ knowledge, no prior study has theoretically examined how the joint distribution of independent Zipfian data leads to the emergence of the Zipf-Mandelbrot law. This topic should be further investigated.

\begin{table}[!ht]
  \centering
  \small
  \caption{$R^2$ scores of Zipf--Mandelbrot fits for the joint-curve model. Rows indicate $s$ and columns indicate $t$.}
  \label{tab:r2_matrix_joint}

  \begin{minipage}{\linewidth}
    \centering
    \textbf{Panel A: Normalized fits.}\par\smallskip
    \begin{adjustbox}{max width=\linewidth}
      \begin{tabular}{r c c c c c c c c}
        \toprule
        $s\backslash t$ & $0.5$ & $1.0$ & $1.5$ & $2.0$ & $2.5$ & $3.0$ & $3.5$ & $4.0$ \\
        \midrule
        0.5 & 0.9990 & 0.9994 & 0.9997 & 0.9998 & 0.9999 & 0.9999 & 0.9999 & 1.0000 \\
        1   & \textemdash & 0.9986 & 0.9989 & 0.9993 & 0.9995 & 0.9997 & 0.9998 & 0.9998 \\
        1.5 & \textemdash & \textemdash & 0.9990 & 0.9993 & 0.9995 & 0.9996 & 0.9997 & 0.9998 \\
        2   & \textemdash & \textemdash & \textemdash & 0.9995 & 0.9996 & 0.9998 & 0.9998 & 0.9999 \\
        2.5 & \textemdash & \textemdash & \textemdash & \textemdash & 0.9998 & 0.9999 & 0.9999 & 0.9999 \\
        3   & \textemdash & \textemdash & \textemdash & \textemdash & \textemdash & 0.9999 & 0.9999 & 1.0000 \\
        3.5 & \textemdash & \textemdash & \textemdash & \textemdash & \textemdash & \textemdash & 1.0000 & 1.0000 \\
        4   & \textemdash & \textemdash & \textemdash & \textemdash & \textemdash & \textemdash & \textemdash & 1.0000 \\
        \bottomrule
      \end{tabular}
    \end{adjustbox}
  \end{minipage}

  \vspace{1em}

  \begin{minipage}{\linewidth}
    \centering
    \textbf{Panel B: Raw fits.}\par\smallskip
    \begin{adjustbox}{max width=\linewidth}
      \begin{tabular}{r c c c c c c c c}
        \toprule
        $s\backslash t$ & $0.5$ & $1.0$ & $1.5$ & $2.0$ & $2.5$ & $3.0$ & $3.5$ & $4.0$ \\
        \midrule
        0.5 & 0.9630 & 0.9954 & 0.9985 & 0.9992 & 0.9996 & 0.9997 & 0.9998 & 0.9998 \\
        1.0 & \textemdash & 0.9987 & 0.9991 & 0.9995 & 0.9996 & 0.9998 & 0.9998 & 0.9999 \\
        1.5 & \textemdash & \textemdash & 0.9993 & 0.9996 & 0.9997 & 0.9998 & 0.9999 & 0.9999 \\
        2.0 & \textemdash & \textemdash & \textemdash & 0.9997 & 0.9998 & 0.9999 & 0.9999 & 0.9999 \\
        2.5 & \textemdash & \textemdash & \textemdash & \textemdash & 0.9999 & 0.9999 & 1.0000 & 1.0000 \\
        3.0 & \textemdash & \textemdash & \textemdash & \textemdash & \textemdash & 1.0000 & 1.0000 & 1.0000 \\
        3.5 & \textemdash & \textemdash & \textemdash & \textemdash & \textemdash & \textemdash & 1.0000 & 1.0000 \\
        4.0 & \textemdash & \textemdash & \textemdash & \textemdash & \textemdash & \textemdash & \textemdash & 1.0000 \\
        \bottomrule
      \end{tabular}
    \end{adjustbox}
  \end{minipage}

\end{table}

\section{Conclusion}\label{sec:Conclusion}
Zipf’s law and its extended version, the Zipf–Mandelbrot law, are commonly observed in a variety of natural data, including natural language and music. In particular, the Zipfian phenomenon in music suggests a semantic structure related to the fundamental compositional nature of music. This indicates that the aesthetic aspects of music perceived by humans may arise from underlying scaling patterns in the data. Indeed, various studies have demonstrated the strong emergence of the Zipf–Mandelbrot law in music data across different time periods and genres, particularly in Western music.

This paper presents two important findings. First, we analyzed Korean traditional music written in the traditional notation system Jeongganbo known as {\it Jeong-ak}, court music. We found that Korean traditional music significantly follows the Zipf–Mandelbrot law, using pitch, duration, and their paired combinations as Zipfian units. For each unit, a clear Zipf–Mandelbrot scaling behavior was observed.
Moreover, this scaling appears to be stronger than that observed in Western music. One possible explanation is that Korean traditional music has been shaped over centuries not only by individual composers but also collectively within society. As a result, the compositional structure may have evolved in a way that is broadly acceptable to members of the community, leading to more pronounced law-consistent patterns.
Our findings provide further evidence that music data exhibit Zipfian properties across cultures and historical periods.

This paper also shows that the joint distribution of two independent Zipfian data sets results in a distribution that follows the Zipf–Mandelbrot law. This may provide an explanation for why the Zipf–Mandelbrot law emerges as an extension of Zipf’s law, since most natural data arise from the joint distribution of two or more Zipfian units.
In fact, this result also applies to our Korean music data. Specifically, pitch and duration each satisfy a piecewise Zipf’s law, and when combined as paired units, they yield a Zipf–Mandelbrot distribution. To the best of the authors’ knowledge, no prior study has theoretically examined how the joint distribution of independent Zipfian data leads to the emergence of the Zipf–Mandelbrot law.

The present study focuses on {\it Jeong-ak} data, court music that represents one of the major categories of Korean traditional music. Another important category is Korean traditional folk music, which is generally more flexible in style and has not always been formally notated. We plan to further investigate the Zipf–Mandelbrot law in this category as well.
We conjecture that the strong manifestation of the Zipf–Mandelbrot law may stem from the fact that, over centuries, idiosyncratic variations have been gradually smoothed out, resulting in musical forms shaped by collective cultural consensus. We will  explore this hypothesis further in our future work.

Regarding the discussion of the joint distribution, our theoretical analysis is limited to the joint distribution of two independent Zipfian data sets, each characterized by a single Zipfian slope. We will extend this research to combinations of piecewise Zipfian data sets.
Moreover, the current work considers only the joint distribution of two Zipfian data sets. Our future research will examine multiple data sets and investigate a more general framework for the Zipf–Mandelbrot law.

\bibliography{reference_main}



\section*{Acknowledgements}
This work was supported by National Research Foundation (NRF) of Korea under the grant number 2021R1A2C3009648 and also partially by KIAS Transdisciplinary Research Program.






\appendix
{\bf \LARGE{Appendix}}
\section{Notation of \textit{yulmyeong} according to the range}

\begin{figure}[!ht]
    \centering
    \includegraphics[width=0.7\linewidth]{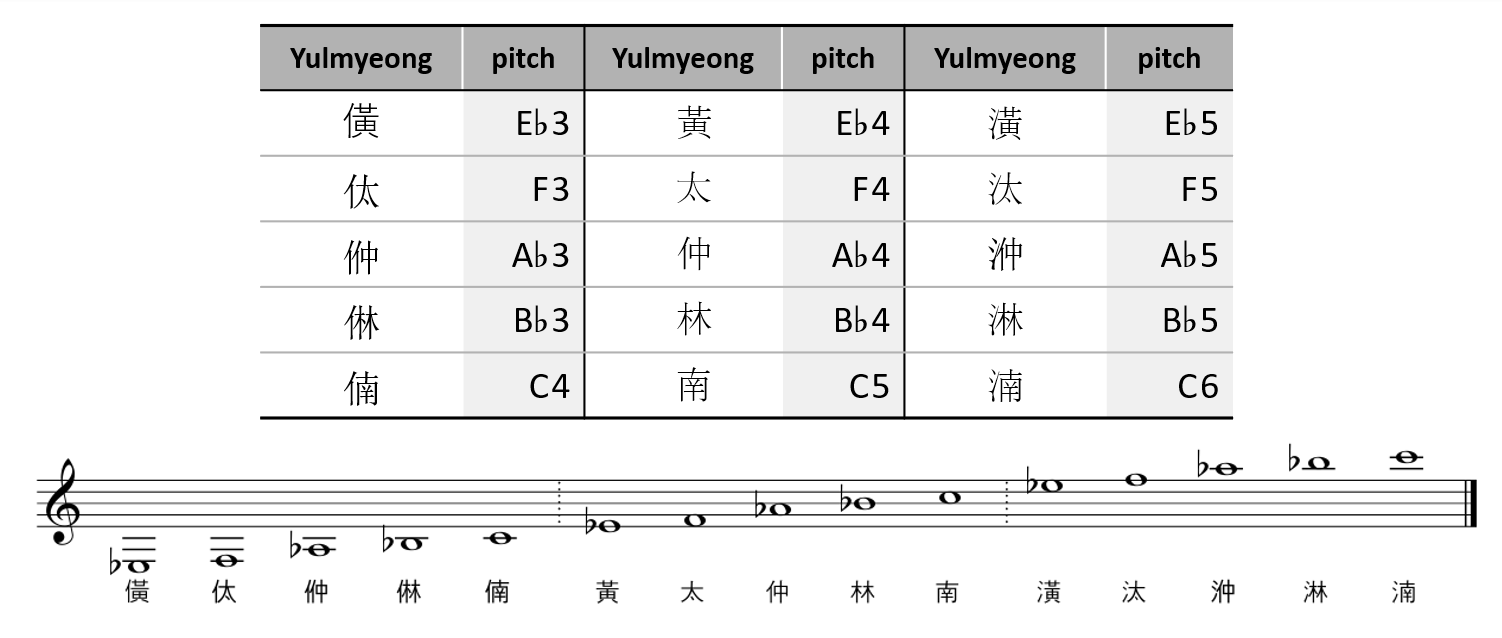}
    \caption{Application of pitch names according to octave}
    \label{fig:yulmyeong table& scale}
\end{figure}

Notation of \textit{yulmyeong} according to the musical range:
Modern-style \textit{jeongganbo} uses a wide range of notes. Therefore, Chinese characters are added to the \textit{yulmyeong} to distinguish the note range. This is commonly used in various instrumental parts.
Each of these distinct note ranges can be converted into Western-style notes.

\section{Mathematical proof of Propositions~\ref{prop:sorting_limit} and~\ref{prop:sorting}}\label{sec:prop_proof}

\begin{proposition}[Sorting converges to the inverse-area curve]
Fix $r>0$ and $\varepsilon_1,\varepsilon_2>0$. Then there exist $M_1,M_2\in\mathbb{N}$ such that
for every $N_1\ge M_1$ and $N_2\ge M_2$ there exists an integer $c\in\mathbb{N}$
for which
\[
\Bigl|\,r-\frac{c}{N_1N_2}\Bigr|<\varepsilon_2
\qquad\text{and}\qquad
\Bigl|\,A^{-1}(r)-f_{*(N_1,N_2)}\Bigl(\frac{c}{N_1N_2}\Bigr)\Bigr|<\varepsilon_1.
\]
In particular, the curve $\{(r,A^{-1}(r))\;|\;r\ge 0\}$ is the continuous analogue of the
rank--frequency plot obtained by sorting lattice samples by decreasing $f$.
\end{proposition}
\begin{proof}
\smallskip
\noindent\textbf{Step 1.}
Fix $f_0\in(0,1]$ and write $\Omega=\Omega(f_0)$.
Set
\[
U:=f_0^{-1/t_1},\qquad V:=f_0^{-1/t_2},
\]
so $\Omega\subset [1,U]\times[1,V]$.

For integers $m,n\in\mathbb N$ define the grid rectangles of mesh
$h_1 \coloneqq \frac1{N_1}$, $h_2 \coloneqq\frac1{N_2}$ by
\[
Q_{m,n}:=\Bigl[\frac{m}{N_1},\frac{m+1}{N_1}\Bigr]\times
          \Bigl[\frac{n}{N_2},\frac{n+1}{N_2}\Bigr].
\]
Consider the outer grid approximation
\[
\Omega^{+}_{N_1,N_2}
:=\bigcup_{\left(\frac{m}{N_1},\frac{n}{N_2}\right)\in \Omega} Q_{m,n}.
\]
Then the rectangles $Q_{m, n}$ in the union have disjoint interiors, each has area
$\frac1{N_1N_2}$, and the index set of above is exactly
$\mathcal L(N_1,N_2)\cap \Omega$.
Hence
\[
\mathrm{Area}(\Omega^{+}_{N_1,N_2})
=\frac{1}{N_1N_2}\#\bigl(\mathcal L(N_1,N_2)\cap\Omega\bigr)
= A_{N_1,N_2}(f_0).
\]
Moreover $\Omega\subset \Omega^{+}_{N_1,N_2}$:
given $(x_1,x_2)\in\Omega$, let $m=\lfloor N_1x_1\rfloor$, $n=\lfloor N_2x_2\rfloor$.
Then $(m/N_1,n/N_2)\le (x_1,x_2)$ coordinatewise, and since $f$ is decreasing
in each coordinate, $f(m/N_1,n/N_2)\ge f(x_1,x_2)\ge f_0$, so
$(m/N_1,n/N_2)\in\Omega$ and therefore $(x_1,x_2)\in Q_{m,n}\subset\Omega^+_{N_1,N_2}$.

Similarly define the inner grid approximation
\[
\Omega^{-}_{N_1,N_2}
:=\bigcup_{\left(\frac{m+1}{N_1},\frac{n+1}{N_2}\right)\in \Omega} Q_{m,n}.
\]
If the upper-right corner of $Q_{m,n}$ lies in $\Omega$, then every point of $Q_{m,n}$
has smaller coordinates, hence larger $f$-value, so $Q_{m,n}\subset\Omega$.
Thus $\Omega^{-}_{N_1,N_2}\subset \Omega\subset \Omega^{+}_{N_1,N_2}$.

Now $\partial\Omega$ is a piecewise $C^1$ curve (the graph
$x_1=f_0^{-1/t_1}x_2^{-t_2/t_1}$ plus segments on $x_1=1$ and $x_2=1$),
hence has finite length; call it $P(f_0)$.
Any rectangle counted in $\Omega^{+}_{N_1,N_2}\setminus\Omega^{-}_{N_1,N_2}$
must intersect $\partial\Omega$, and it is contained in the $\sqrt{2}h$-neighborhood
of $\partial\Omega$ with $h:=\max\{h_1,h_2\}$.
By Kneser's theorem for the area of tubular neighborhoods~\cite{kneser1955einige},
the area of this $\sqrt{2}h$-neighborhood is $O(h)$; concretely,
\[
0\le \mathrm{Area}(\Omega^{+}_{N_1,N_2})-\mathrm{Area}(\Omega^{-}_{N_1,N_2})
\le 2 P(f_0)\,h + O(h^2)\xrightarrow[N_1,N_2\to\infty]{}0.
\]
Since $\Omega^{-}_{N_1,N_2}\subset\Omega\subset\Omega^{+}_{N_1,N_2}$, we get
\[
\mathrm{Area}(\Omega^{+}_{N_1,N_2})\to \mathrm{Area}(\Omega)=A(f_0),
\quad\text{i.e.}\quad
A_{N_1,N_2}(f_0)\to A(f_0).
\]

\smallskip
\noindent\textbf{Step 2.}
Let $\tau:=A^{-1}(r)$. Fix any $\delta\in(0,\varepsilon_1)$.
Because $A$ is strictly decreasing and continuous,
\[
A(\tau+\delta)<A(\tau)=r<A(\tau-\delta).
\]
Define the positive margin
\[
\eta:=\frac14\min\bigl\{\,r-A(\tau+\delta),\;A(\tau-\delta)-r\,\bigr\}>0.
\]
By Step~1, choose $M_1,M_2$ so that for all $N_1\ge M_1$, $N_2\ge M_2$,
\[
\bigl|A_{N_1,N_2}(\tau\pm\delta)-A(\tau\pm\delta)\bigr|<\eta.
\]
Then
\[
A(\tau + \delta ) \le r - 4 \eta \quad \text{and} \quad
A_{N_1, N_2}(\tau + \delta) < (r - 4 \eta) + \eta,
\]
hence
\[
A_{N_1,N_2}(\tau+\delta)\le r-2\eta.\]
Similarly,
\[
A_{N_1,N_2}(\tau-\delta)\ge r+2\eta.
\]
Next enlarge $M_1,M_2$ further so that $\frac{1}{2N_1N_2}<\min\{\varepsilon_2,\eta\}$.
Pick $c\in\mathbb N$ to be the nearest integer to $rN_1N_2$.
Then
\[
\Bigl|r-\frac{c}{N_1N_2}\Bigr|<\varepsilon_2
\quad\text{and}\quad
r-\eta<\frac{c}{N_1N_2}<r+\eta,
\]
so in particular
\[
A_{N_1,N_2}(\tau+\delta)<\frac{c}{N_1N_2}<A_{N_1,N_2}(\tau-\delta).
\]

Finally use the interpretation of $A_{N_1,N_2}(\cdot)$ as the normalized counting functional of the superlevel set.
The inequality
$A_{N_1,N_2}(\tau+\delta)<\frac{c}{N_1N_2}$
means that there are fewer than $c$ lattice points whose values are greater than or equal to $\tau + \delta$, hence the $c$-th largest value is \emph{smaller than} $\tau+\delta$.
Likewise $A_{N_1,N_2}(\tau-\delta)>\frac{c}{N_1N_2}$
means that there are more than $c$ lattice points whose values are greater than or equal to $\tau - \delta$, hence the $c$-th largest value is \emph{at most} $\tau-\delta$.
Therefore
\[
\tau-\delta
\;\le\;
f_{*(N_1,N_2)}\Bigl(\frac{c}{N_1N_2}\Bigr)
\;\le\;
\tau+\delta,
\]
and recall that $\tau = A^{-1}(r)$,
\[
\Bigl|A^{-1}(r)-f_{*(N_1,N_2)}\Bigl(\frac{c}{N_1N_2}\Bigr)\Bigr|
\le \delta<\varepsilon_1.
\]
Together with $\bigl|r-\frac{c}{N_1N_2}\bigr|<\varepsilon_2$, this completes the proof.
\end{proof}

\begin{proposition}[Closed form of the sorting function]
Fix $t_1,t_2>0$. For $f_0\in(0,1]$, define the superlevel set
\[
\Omega(f_0)\;:=\;\bigl\{(x_1,x_2)\in[1,\infty)^2:\;x_1^{-t_1}x_2^{-t_2}\ge f_0\bigr\},
\]
and let $A(f_0):=\operatorname{Area}\,\Omega(f_0)$. Then:
\begin{enumerate}
\item If $t_1\neq t_2$,
\[
A(f_0)\;=\;1+\frac{t_1\,f_0^{-1/t_1}-t_2\,f_0^{-1/t_2}}{\,t_2-t_1\,}.
\]
\item If $t_1=t_2=t$,
\[
A(f_0)\;=\;f_0^{-1/t}\!\Bigl(\tfrac{1}{t}\ln\tfrac{1}{f_0}-1\Bigr)+1.
\]
\end{enumerate}
Moreover, $A:(0,1]\to[1,\infty)$ is strictly decreasing and hence bijective; the
continuous rank--frequency curve is the inverse
\[
f_*(r)\;=\;A^{-1}(r)\qquad(r\ge 1).
\]
\end{proposition}

\begin{proof}
From $x_1^{-t_1}x_2^{-t_2}\ge f_0$ and $x_1,x_2\ge 1$ we obtain
\[
x_1 \;\le\; f_0^{-1/t_1}\,x_2^{-t_2/t_1}.
\]
Thus for each fixed $x_2\ge1$ the admissible $x_1$–interval is
\[
x_1\in\Bigl[\,1,\;f_0^{-1/t_1}x_2^{-t_2/t_1}\Bigr],
\]
which contributes positive width only when $x_2\le f_0^{-1/t_2}$. Hence
\[
\Omega(f_0)=\Bigl\{(x_1,x_2):\;1\le x_2\le f_0^{-1/t_2},\;\;1\le x_1\le f_0^{-1/t_1}x_2^{-t_2/t_1}\Bigr\},
\]
and
\[
A(f_0)=\int_{1}^{\,f_0^{-1/t_2}}\Bigl(f_0^{-1/t_1}x_2^{-t_2/t_1}-1\Bigr)\,dx_2.
\]

\noindent
If $t_1\neq t_2$, writing $\alpha:=t_2/t_1\neq 1$,
\[
\begin{aligned}
A(f_0)
&= f_0^{-1/t_1}\,\frac{x_2^{\,1-\alpha}}{\,1-\alpha\,}\Big|_{1}^{\,f_0^{-1/t_2}}
\;-\; \bigl(x_2\big|_{1}^{\,f_0^{-1/t_2}}\bigr)\\[1mm]
&= 1+\frac{t_1\,f_0^{-1/t_1}-t_2\,f_0^{-1/t_2}}{\,t_2-t_1\,}.
\end{aligned}
\]

\noindent
If $t_1=t_2=t$, then the upper bound is $f_0^{-1/t}x_2^{-1}$ and
hence
\[
\begin{aligned}
A(f_0)
&=\int_{1}^{f_0^{-1/t}}\Bigl(f_0^{-1/t}x_2^{-1}-1\Bigr)\,dx_2\\
&= f_0^{-1/t}\,\ln\!\bigl(f_0^{-1/t}\bigr)-\bigl(f_0^{-1/t}-1\bigr)
= f_0^{-1/t}\!\Bigl(\tfrac{1}{t}\ln\tfrac{1}{f_0}-1\Bigr)+1.
\end{aligned}
\]

\noindent
Finally, differentiating either expression in $f_0$ shows $A'(f_0)<0$, so $A$ is strictly decreasing on $(0,1]$ and admits a well-defined inverse $A^{-1}$, which yields the continuous rank–frequency curve $f_*(r)=A^{-1}(r)$.
\end{proof}

\end{document}